\documentclass[a4paper, UKenglish, cleveref, autoref, thm-restate]{lipics-v2021}
\nolinenumbers
\usepackage{verbatim}
\usepackage{dsfont}
\usepackage{tikz}
\usetikzlibrary{trees, arrows, patterns, shapes, snakes, fit, shadows, calc, automata, decorations.markings, calc, positioning}
\usepackage{todonotes}
\usepackage{mathtools}
\usepackage{amsmath}
\usepackage{subcaption}
\usepackage{graphicx}
\usepackage{float}

\newcommand\pr{{Proof}}
\newcommand\alt{{Alt}}
\newcommand\pred{{Pred}}
\newcommand\struc{\mathbb{A}}
\newcommand\nat{\mathbb{N}}
\newcommand\stack{\mathds{P}}
\newcommand\op{\mathcal{O}p}
\newcommand\Pred{\mathcal{P}red}
\newcommand\aut{\mathcal{A}}
\newcommand\gram{\mathcal{G}}
\newcommand\lm{LmD}

\title{On the equivalence between generating functions computed by memory transducers and enumerating functions produced by indexed grammars} 

\titlerunning{} 

\author{Vincent Ghigo}{Université de Bordeaux, CNRS, Bordeaux INP, LaBRI, UMR 5800, Talence, 33400, France}{vincent.ghigo@labri.fr}{https://orcid.org/0000-0002-1825-0097}{}

\authorrunning{V. Ghigo } 

\Copyright{Vincent GHIGO} 

\ccsdesc[500]{Theory of computation~Formal languages and automata theory}

\keywords{Indexed Grammar, Pushdown Transducer, Enumeration problem} 

\category{} 

\relatedversion{} 

\acknowledgements{I want to thank my PhD advisors V.Penelle and G.Sénizergues}

\EventEditors{John Q. Open and Joan R. Access}
\EventNoEds{2}
\EventLongTitle{42nd Conference on Very Important Topics (CVIT 2016)}
\EventShortTitle{CVIT 2016}
\EventAcronym{CVIT}
\EventYear{2016}
\EventDate{December 24--27, 2016}
\EventLocation{Little Whinging, United Kingdom}
\EventLogo{}
\SeriesVolume{42}
\ArticleNo{23}

\begin{document}

\maketitle

\begin{abstract}
We consider the sequences of natural integers that can be \emph{computed} by a deterministic transducer, with input in a structure $\struc$, output in $\nat$
and with memory the set of stacks of stacks of $\struc$. We show that these sequences are, exactly, the \emph{counting sequences} of formal languages \emph{generated} by unambiguous context-free indexed grammars (equivalently, the counting sequences of derivation trees of arbitrary context-free indexed grammars), 
with indexes in $\struc$. This general theorem applies, notably, to the set of natural integers endowed with the operation $n \mapsto n\stackrel{.}{-}1$ and the non-zero predicate, showing that the polynomial recurrences count exactly the \emph{index}-languages (where the parameter used for counting is the index itself).
\end{abstract}
\newpage
\section{Introduction}\label{sec:intro}
\paragraph*{The Context}
Relations between automata and grammars is a well-studied topic, in particular for comparing the class of languages they define \cite{Chomsky59}.\\
This work establishes a similar relation in order to contribute to answering two kinds of questions:\begin{description}
\item[Counting:]\label{test} Given a set $E(n)$ of combinatorial objects (words, trees, planar maps, etc ...) depending on the integer $n \in \nat$, what is the \emph{number} $u(n)$ of these objects? How can it be computed? by which kind of recurrence? or by which kind of transducer?
\item[Representation:] Given a sequence $u(n)$ of natural integers (defined by a recurrence or a transducer), can we represent $u(n)$ as the \emph{number} of objects of some set $E(n)$ of combinatorial objects?
\end{description}
Counting problems are central in combinatorics; therefore, they  have been widely studied  for many kinds of objects. 
Representation problems have well-known solutions for $\nat$-linear sequences and $\nat$-algebraic series; recently, they have been tackled in \cite{DBLP:journals/fuin/KotekM12,DBLP:journals/ita/AdamsFM13} stimulated by a question about D-finite series raised in (\cite{DBLP:conf/stacs/Bousquet-Melou05}, section 6).

\paragraph*{The result}
The main results of this article are \cref{thm:gram_to_trans} and \cref{thm:trans_to_gram}, which answer these two questions in the following context:
\begin{itemize}
    \item The counting problem is asked for the set $E(\sigma)$ of all trees generated by an indexed grammar, where indexes $\sigma$ belong to a structure $\struc$.
    \item The representation problem is asked for the sequences $u(\sigma)$ computed by a deterministic transducer whose input belongs to a structure $\struc$, 
    whose output lies in $\nat$, and whose memory consists of the set of stacks of stacks over $\struc$.
\end{itemize}
\paragraph*{The framework}
We place ourselves in a context in which we study computation models over an arbitrary structure $\struc$, following \cite{DBLP:journals/iandc/Engelfriet91}. 
This allows us to both give general theorems and use techniques that shed light on computation models themselves. 

We define  a transducer with memory over a structure $\struc$, as an automaton with memory in $\struc$ whose initial memory is given as input.

An indexed grammar on $\struc$, as introduced in \cite{DBLP:journals/jacm/Aho68}, is a context-free grammar in which each non-terminal is indexed by an element of $\struc$, and rules may depend on this element via guards; these rules can also apply operations of $\struc$ to the index.

As a consequence of the generalization to an arbitrary structure, we also give an answer to the representation problem for transducers with the set of $k$-stacks of $\struc$ as memory for all $k\geq2$.

\section{Preliminaries}\label{sec:preliminaries}
    \begin{description}
        \item[Structure]\label{desc:struct}
            A \emph{structure} is a tuple $<\mathcal{S}, \Pred, \op>$ where:
            \begin{itemize}
                \item $\mathcal{S}$ is a \emph{set of elements}. We denote by $\sigma$ an element of $\mathcal{S}$.
                \item $\Pred$ is a \emph{set of predicates} $\mathcal{S} \to \{\top, \bot\}$.
                \item $\op$ is a \emph{set of operations} on $\mathcal{S} \to \mathcal{S}$. We denote by $op$ an element of $\op$.
            \end{itemize}
        \item[Stack]\label{desc:def_stack}
            Let $\struc = \, <\mathcal{S}_{\struc}, \Pred_{\struc}, \op_{\struc}>$ be a structure and $\Gamma$ a finite set of symbols.\\
            A \emph{stack structure} $\stack(\Gamma, \struc)$ is a structure where :
            \begin{itemize}
                \item ${(\Gamma \times \mathcal{S}_{\struc})}^*$ is the \emph{set of elements}. We note $X[\sigma]$ an atom $(X,\sigma)\in(\Gamma \times \mathcal{S}_{\struc})$.
                \item $\Pred_{\struc} \cup \{\mathcal{P}_X \mid X \in \Gamma \}$ is the \emph{set of predicates}, with $\mathcal{P}_X(stk)$ being true if the topmost symbol of $stk$ is $X$.
                \item $\op_{\struc} \cup \{push(\gamma) \mid \gamma \in \Gamma^*\}$ is a \emph{set of operations}.
            \end{itemize}
            We denote by $X, Y, Z$ symbols of $\Gamma$, $stk$ an element of ${(\Gamma \times \mathcal{S}_\struc)}^*$, $X[\sigma]$ an atom in ${(\Gamma \times \mathcal{S}_\struc)}$ and $\sigma$ an index in $\mathcal{S}_\struc$. 
            Operations and predicates from $\struc$ extend to the  stack structure by acting/evaluating on the topmost atom of the stack. 
            More formally, we define the application $op$ to a stack as $op(X[\sigma]stk):= X[op(\sigma)]stk$ and the evaluation of $p\in\Pred_\struc$ on a stack as $p(A[\sigma]stk):= p(\sigma)$.\\
            Let $\gamma = \prescript{1}{}{X} \prescript{2}{}{X} \ldots \prescript{k}{}{X} \in \Gamma^*$. 
            The operation $push(\gamma)$ on $X[\sigma]stk$ erases the topmost atom $X[\sigma]stk$ and copy the index $\sigma$ in each symbol of $\gamma$.
            More formally, $push(\gamma)$ is defined as $push(\gamma)(X[\sigma]stk):= \prescript{1}{}{X}[\sigma]\prescript{2}{}{X}[\sigma] \ldots \prescript{k}{}{X}[\sigma]stk$.
            For the sake of clarity, we will note $push(\varepsilon)$ as $pop$.\\
            When considering nested stacks, for the sake of clarity, we index stack symbols according to their \emph{order}, with the outermost stack designated as order $1$.
            We will also label the stack operations by the order. 
            We write $\stack^k(\struc)$ for the order $k$ stacks on $\struc$ when the alphabets are not relevant.\\
            In this article, we consider stacks with \emph{bottom-of-stack symbol} $\#$ and the predicate $\mathcal{P}_\#$.
            \begin{figure}[h]
                \centering
                \begin{subfigure}{0.39\textwidth}
                \begin{tikzpicture}
                    \draw (0,0)node (q0) [anchor=west]{$X[\prescript{1}{}{\sigma}]Y[\prescript{2}{}{\sigma}]$} ;
                    \draw (5,0)node (q1) [anchor=west]{$Y[\prescript{1}{}{\sigma}]Z[\prescript{1}{}{\sigma}]Y[\prescript{2}{}{\sigma}]$} ;
                    \draw (5,1)node (q2) [anchor=west]{$X[op(\prescript{1}{}{\sigma})]Y[\prescript{2}{}{\sigma}]$} ;
                    \draw (5,-1)node (q3) [anchor=west]{$Y[\prescript{2}{}{\sigma}]$} ;
                    \draw [->](q0) to  node [align=center] (note1) [midway, above] {$push(YZ)$}(q1);
                    \draw [->](q0) to [bend left,in=160] node [align=center] (note1) [midway, above] {$op$}(q2);
                    \draw [->](q0) to [bend right,in=193] node [align=center] (note1) [midway, above] {$pop$}(q3);
                \end{tikzpicture}\caption{Visualization of operation on stack}\label{subfig:1}
                \end{subfigure}
                \hfill
                \begin{subfigure}{0.35\textwidth}
                \begin{tikzpicture}
                    \draw [-](0,2.5)--(0,0)--(4.5,0)--(4.5,2.5);
                    \draw [dashed](0,2.5)--(0,3);
                    \draw [dashed](4.5,2.5)--(4.5,3);
                    \draw (0.5,0.5)node (q2) [anchor=west]{$X_2[\prescript{1}{}{\sigma}]Y_2[\prescript{2}{}{\sigma}]Z_2[\prescript{3}{}{\sigma}]$} ;
                    \draw [-](0.6,0.2)--(4.35,0.2)--(4.35,0.8)--(0.6,0.8)--(0.6,0.2);
                    \draw (0,0.5)node (q2) [anchor=west]{$X_1$} ;
                    \draw (0.5,1.2)node (q2) [anchor=west]{$Z_2[\prescript{3}{}{\sigma}]X_2[\prescript{4}{}{\sigma}]$} ;
                    \draw [-](0.6,0.9)--(4.35,0.9)--(4.35,1.5)--(0.6,1.5)--(0.6,0.9);
                    \draw (0,1.2)node (q2) [anchor=west]{$Y_1$} ;
                    \draw (0.5,1.9)node (q2) [anchor=west]{$Y_2[\prescript{2}{}{\sigma}]X_2[\prescript{5}{}{\sigma}]Z_2[\prescript{6}{}{\sigma}]$} ;
                    \draw [-](0.6,1.6)--(4.35,1.6)--(4.35,2.2)--(0.6,2.2)--(0.6,1.6);
                    \draw (0,1.9)node (q2) [anchor=west]{$X_1$} ;
                \end{tikzpicture}
                \caption{Visualization of stack of stack }\label{subfig:2}
            \end{subfigure}
            \end{figure}
        \item[Memory transducer]\label{desc:def_trans}
            An \emph{$\struc$-transducer} is a tuple $\aut:= (\mathcal{Q}, T, \Delta, q_0, \struc, Start, End)$, such that:
            \begin{itemize}
                \item $\mathcal{Q}$ is the finite \emph{set of states} and $q_0 \in \mathcal{Q}$ is the initial state. We denote by $q, p, r, s$  elements of $\mathcal{Q}$.
                \item $T$ is the \emph{set of terminal symbols}. We denote by $a$ an element of $T$ and $\bar{a}$ an element of $T \cup \{\varepsilon\}$. 
                \item $Start \subseteq \mathcal{S}_\struc$ is the \emph{set of initial elements}.  
                \item $\Delta \subseteq \bigl((\mathcal{Q} \times \mathcal{F}) \times (\mathcal{Q} \times T \cup \{\varepsilon\} \times \op_\struc)\bigr)$ is the \emph{set of rules}, where $\mathcal{F}$ is the set of boolean formulae over  $\pred_\struc$.
                We note $(q,g)\to(p,\bar{a},op)$ a rule $\bigl((q,g),(p,\bar{a},op)\bigr)\in \Delta$. 
                \item $End \in \mathcal{Q} \times \mathcal{S}_\struc$ is the \emph{accepting configuration}.
            \end{itemize}
            \begin{figure}[h]
            \begin{center}
            \begin{tikzpicture}[shorten >= 1pt, node distance = 3.5cm, on grid, auto]
                \node [state,shape=ellipse, align=center, double] (q0)  {$q_0$}; 
                \node [] (q-1)[below = 1cm of q0]  {}; 
                \node [state,shape=ellipse, align=center] (q1) [right=of q0] {$q_1$}; 
                \node [state,shape=ellipse, align=center] (q4) [left=of q0] {$q_3$}; 
                \node [state,shape=ellipse, align=center] (q3) [below=of q0] {$q_2$}; 
                \draw [->] (q0) to [bend left] node [align=center] (note1) [midway, above] {$(\mathcal{P}_{A_2}\wedge \neq 0,\varepsilon,-1)$}(q1);
                \draw [->] (q1) to [bend left] node [align=center] (note1) [midway, right] {$(\top,\varepsilon,push_2(B_2A_2))$}(q3);
                \draw [->] (q0) to [bend left] node [align=center] (note1) [midway, left] {$(\mathcal{P}_{B_2}\wedge \neq 0,\varepsilon,-1)$}(q3);
                \draw [->] (q3) to [bend left] node [align=center] (note1) [midway, left] {$(\top,\varepsilon,push_1(A_1A_1))$}(q4);
                \draw [->] (q4) to [bend left] node [align=center] (note1) [midway, above] {$(\top,\varepsilon,pop_2)$}(q0);
                \draw [->] (q0) to [out=105, in=75, looseness=10] node [align=center] (note1) [midway,above ] {$(\mathcal{P}_{A_2}\wedge = 0,a,pop_1)$\\$(\mathcal{P}_{B_2}\wedge = 0,\varepsilon,pop_2)$\\}(q0);
                \draw [->] (q-1) to [] node [align=center]  [midway, above] {}(q0);
            \end{tikzpicture}  
                \caption{Example of $\stack^2(\nat)$-transducer computing $(i+1)!$}\label{fig:trans}
            \end{center}
        \end{figure}
            On $\aut$, we use the following semantics: 
            \begin{itemize}
                \item the one-step relation $\xrightarrow{\delta}$ with $\delta = \bigl((q, g_\delta) \to (p, \bar{a}_\delta, op_\delta)\bigr) \in \Delta$ is $\Bigl\{\bigl((q, \sigma), (p, \sigma')\bigr) \mid \sigma' = op_\delta(\sigma) \wedge \sigma \vDash g_\delta\Bigr\}$. 
                \item $\xrightarrow{\varepsilon}$ is the identity relation. 
                \item $\xrightarrow{\delta.c}:= \xrightarrow{\delta}.\xrightarrow{c}$ with $c \in \Delta^*$.
            \end{itemize}
            We call a \emph{run} $(q, \sigma) \xrightarrow{c} (p, \sigma')$:
            \begin{itemize}
                \item \emph{accepting} if $(p, \sigma') = End$. 
                \item If $\struc = \stack^m(\mathbb{B})$, for some $\mathbb{B}$ with $m \geq 1$:
                \begin{itemize}
                    \item \emph{well-nested} if operations of order $1$ in $c$ form a Dyck word, i.e. there is an equal number of $push_1$ and $pop_1$ in $c$, and for every prefix of $c$ there are more $push_1$ than $pop_1$. 
                    \item \emph{well-formed} if $c = c'.\delta$ with $op_\delta = pop_1$ and $c'$ is \emph{well-nested}. 
                \end{itemize}
            \end{itemize}
            We assume that $End\xrightarrow{c}(p,\sigma)$ if and only if $c=\varepsilon$.\\
            We call a $\struc$-transducer \emph{computing} if it is:
            \begin{itemize}
                \item \emph{Deterministic}: $\forall (q, \sigma) \in \mathcal{Q} \times \mathcal{S}_\struc, \exists^{\leq 1} \delta \in \Delta, \exists (p, \sigma') \in \mathcal{Q} \times \mathcal{S}_\struc, (q, \sigma) \xrightarrow{\delta} (p, \sigma')$.
                \item \emph{Total}: $\forall \sigma \in Start, (q_0, \sigma) \xrightarrow{*} End$.
                \item $Card(T) = 1$.
            \end{itemize}
            We say that $\aut$, with $T_{\aut}=\{a\}$, \emph{computes} $u:\mathcal{S}_\struc\to\nat$ if $\aut$ is \emph{computing} and $\forall \sigma \in Start, (q_0, \sigma) \xrightarrow{c} End \wedge \underset{\delta \in c}{\odot} \bar{a}_{\delta} = a^{u(\sigma)}$, where $\odot$ is a notation for concatenation.
        \item[Tree]\label{desc:def_tree}
            Let $\mathcal{S}$ be a set. We define inductively $\mathcal{T}(\mathcal{S})$, the set of \emph{$\mathcal{S}$-labelled trees}, as the set of tuples $(\sigma, l)$ where $l \in {\mathcal{T}(\mathcal{S})}^*$ is the list of subtrees and $\sigma \in \mathcal{S}$ is the label of the root. 
            If $l$ is empty, we call $(\sigma, l)$ a \emph{leaf}. For a tree $t:= (\sigma, l)$, we call the label of the root $rt(t):= \sigma$. For a list of trees $l = t_1, t_2, \ldots, t_k$, $rt(l):= rt(t_1)rt(t_2)\ldots rt(t_k)$.
            For a tree $t = (\sigma, l)$, we define the set of labels of the leaves $leaf(t):= \left\{
            \begin{array}{ll}
            \sigma & \mbox{if } l = \emptyset\\
            \underset{t' \in l}{\bigcup} leaf(t') & \mbox{else}
            \end{array}
            \right.$
        \item[Indexed grammar]\label{desc:def_gram}
            An \emph{$\struc$-grammar scheme} is a tuple $\gram_{scheme}:= (N, T, \struc, P)$, such that: 
            \begin{itemize}
                \item $N$ is the \emph{set of non-terminal symbols}. We denote by $A, B, C$ elements of $N$.
                \item $T$ is the \emph{set of terminal symbols}. We denote by $a$ an element of $T$.
                \item $\struc$ is a structure.
                \item $P \subseteq \Bigl((N \times \mathcal{F}) \times \bigl((T \cup N)^* \times \op_\struc\bigr)\Bigr)$ is the \emph{set of rules}, 
                where $\mathcal{F}$ is the set of boolean formulae over $\pred_\struc$. 
                We note $\bigl((A,g)\to(\gamma,op)\bigr)$ a rule $\bigl((A,g),(\gamma,op)\bigr)\in P$.
            \end{itemize}
            We define the \emph{derivation trees} of a $\struc$-grammar scheme as:
            \begin{itemize}
                \item Let $r = (A, \mathrm{g}) \to (\prescript{0}{}{\beta} \prescript{1}{}{A}  \prescript{1}{}{\beta}  \prescript{2}{}{A} \ldots \prescript{k}{}{A} \prescript{k}{}{\beta} , op) \in P$, with $\prescript{i}{}{\beta}  \in (T^*)$ and $\prescript{i}{}{A} \in N$. \\
                The one-step derivation $B[\sigma] \to Z$ is accepted by $r$ if, and only if, $B = A$, $\sigma \vDash \mathrm{g}$, and $Z = \prescript{0}{}{\beta} \prescript{1}{}{A}[op(\sigma)] \prescript{1}{}{\beta} \prescript{2}{}{A}[op(\sigma)] \ldots \prescript{k}{}{A}[op(\sigma)] \prescript{k}{}{\beta}$. 
                \item A \emph{derivation tree} of $\gram$ is inductively defined by:
                \begin{itemize}
                    \item $(a, \emptyset)$.
                    \item $(A[\sigma], l)$ for $l$ a list of \emph{derivation tree} such that $A[\sigma] \to rt(l)$ is accepted by a rule of $\gram_{scheme}$.
                \end{itemize}
            \end{itemize}
            We define $\gram$ an \emph{$\struc$-grammar} as a tuple $(N, T, \struc, P, S)$ such that $(N, T, \struc, P)$ is a $\struc$-grammar scheme and $S \in N$ that we call the \emph{axiom}.
            A \emph{derivation tree} $t$ of $\gram_{scheme}$ is a derivation tree of $\gram$ if $\exists \sigma \in \mathcal{S}_\struc, rt(t) = S[\sigma]$.
            \begin{figure}
                \centering
    \begin{tabular}{lll|lll}
        $(q_0,\mathcal{P}_{A}\wedge = 0)$&$\to$&$(a,id)$& 
        $(q_3,\top)$ &$\to$&$(q_0,pop)$\\
        $(q_0,\mathcal{P}_{A}\wedge\neq 0)$ &$\to$ &$(q_1,-1)$&
        $(q_1,\top)$ &$\to$&$(q_2,push(BA))$\\
        $(q_0,\mathcal{P}_{B}\wedge \neq 0)$ &$\to$&$(q_2,-1)$&
        $(q_2,\top)$ &$\to$&$(q_3,id)$\\
        $(q_0,\mathcal{P}_{B}\wedge=0)$&$\to$&$(q_0,pop)$\\
    \end{tabular}\\
    \begin{tikzpicture}[level distance=2cm,
                    sibling distance=3cm]
                    \node {${q_0}_{[A[0]\dots]}$}
                        child {
                            node {$a$}
                        };
    \end{tikzpicture}
                \vline 
    \begin{tikzpicture}[level distance=0.8cm,
                    sibling distance=4.5cm]
                    \node {${q_0}_{[A[i]\dots]}$}
                        child {
                            node {${q_1}_{[A[i-1]\dots]}$} 
                              child {
                                node {${q_2}_{[B[i-1]A[i-1]\dots]}$} 
                                  child {
                                    node {${q_3}_{[B[i-1]A[i-1]\dots]}$} 
                                      child {
                                        node {${q_0}_{[A[i-1]\dots]}$} 
                                          child [level distance=0.7cm]{node [draw,dashed,shape border uses incircle,isosceles triangle,shape border rotate=90, minimum height=5mm] {}}
                                      }
                                  }
                                  child {
                                    node {${q_0}_{[B[i-1]A[i-1]\dots]}$} 
                                      child {
                                        node {${q_2}_{[B[i-2]A[i-1]\dots]}$} 
                                          child {
                                            node {${q_3}_{[B[i-2]A[i-1]\dots]}$} 
                                              child {
                                                node {${q_0}_{[A[i-1]\dots]}$} 
                                                  child [level distance=0.7cm]{node [draw,dashed,shape border uses incircle,isosceles triangle,shape border rotate=90, minimum height=5mm] {}}
                                              }
                                          }
                                          child{
                                            node {${q_0}_{[B[i-2]A[i-1]\dots]}$} 
                                            child [level distance=0.7cm]{node [draw,dashed,shape border uses incircle,isosceles triangle,shape border rotate=90, minimum height=5mm] {}}
                                          }
                                      }
                                  }
                              }
                        };
                \end{tikzpicture}
                \caption{Example of an $\stack(\nat)$-grammar and derivation tree}
            \end{figure}
        \item[Alternating transducer]\label{desc:def_alt_trans}
            We will use alternating transducers as a technical object to make the transformation from transducer to grammar easier and more understandable.
            An \emph{alternating $\stack(\struc)$-transducer} $\aut_{\alt}:= (\mathcal{Q}, T, \Delta, q_0, \stack(\struc), Start, End)$ is a $\stack(\struc)$-transducer where transitions lead to conjunctions\footnote{Rigorously, conjontions are taken in free distributive lattice on $\mathcal{Q}\times(T\cup\varepsilon)\times \op$} of a state, a non-terminal or $\varepsilon$, and an operation. 
            \begin{itemize}
                \item Let $\delta = (q, \mathrm{g}) \to (\prescript{1}{}{q}, \prescript{1}{}{\bar{a}}, \prescript{1}{}{op}) \wedge \ldots \wedge (\prescript{k}{}{q}, \prescript{k}{}{\bar{a}}, \prescript{k}{}{op}) \in \Delta$. 
                The one-step relation $\xrightarrow{\delta}$ is
                $\Bigl\{\bigl((q, \sigma), \varphi\bigr) \mid (\sigma \vDash \mathrm{g}) 
                \wedge \bigl(\varphi = (\prescript{1}{}{q}, \prescript{1}{}{\sigma}) 
                \wedge \ldots \wedge (\prescript{k}{}{q}, \prescript{k}{}{\sigma})\bigr) 
                \wedge \bigl(\bigwedge_{1 \leq i \leq k} \prescript{i}{}{\sigma} = \prescript{i}{}{op}(\sigma)\bigr)\Bigr\}$. 
                \item A \emph{run tree} of $\aut_{\alt}$ is, inductively, defined by:
                \begin{itemize}
                    \item $\bigl((q, stk), \emptyset\bigr)$.
                    \item $\bigl((q, stk), l\bigr)$ with $l$ a \emph{run tree} such that $rt(l) = (\prescript{1}{}{q}, \prescript{1}{}{stk}), \ldots, (\prescript{k}{}{q}, \prescript{k}{}{stk})$ 
                    and there exists $\delta \in \Delta$ for which $(q, stk) \xrightarrow{\delta} \bigwedge_{1 \leq i \leq k} (\prescript{i}{}{q}, \prescript{i}{}{stk})$.  
                \end{itemize} 
                \item An \emph{accepting run tree} of $\aut_{\alt}$ is defined as a run tree where all leaves are $(End, \emptyset)$.
                \item We write $(q, stk) \xrightarrow{\alpha} \bigwedge_{1 \leq i \leq k} (\prescript{i}{}{q}, \prescript{i}{}{stk})$  with $\alpha \in \mathcal{T}(\Delta)$ 
                if there exists a run tree $t$ of $\aut_{\alt}$ rooted in $(q, stk)$ such that $leaf(t) = \bigl\{(\prescript{1}{}{q}, \prescript{1}{}{stk}), \ldots, (\prescript{k}{}{q}, \prescript{k}{}{stk})\bigr\}$, 
                the trees $\alpha$ and $t$ are equivalent up to relabelling and the rule used from a node in $t$ to its children is the label of the corresponding node in $\alpha$. 
                \item We call a run tree \emph{well-formed} if all of its branches are \emph{well-formed} runs. 
            \end{itemize}
        \item[Symbolic runs]\label{desc:def_symb_run}
            Observe that stack operations and predicates only touch the topmost of the stack. 
            Because of this, in transducers and alternating transducers, if there is a run (resp. run tree) from $(q, stk)$, then for every stack $r$, 
            there is a run (resp. run tree) with the same transitions from $(q, stk.r)$, as $r$ is never observed by said run. 
            To deal with this fact, in transducers and alternating transducers, when dealing with sub-runs of the principal run, 
            we use the notion of so-called \emph{symbolic runs} to represent sets of runs with the same transitions (resp. tree), with a symbol $\Omega$ to represent the bottom of the stack that is never observed. 
            This notion is crucial to our proofs, as we show that derivation trees of the grammars we construct are in bijection with symbolic runs of transducers, rather than with concrete runs.
    \end{description}
\section{From indexed grammar to memory transducer}\label{sec:gram_to_trans}
    In this part, we prove the following theorem:
    \begin{theorem}\label{thm:gram_to_trans}
    Let $\gram$ be a $\struc$-indexed grammar, without infinite derivations, producing $u(\sigma)$ derivation trees from the axiom indexed by $\sigma$ for all $\sigma \in \mathcal{S}_\struc$. 
    There exists a $\stack^2(\struc)$-transducer that \emph{computes} $u(\sigma)$.
    \end{theorem}
    \subsection{Grammar normal form}\label{subsec:gram_normal_form}
    Without loss of generality, we consider a $\struc$-grammar in a normal form similar to Chomsky normal form for context-free grammars. 
    This normal form separates terminal production, non-terminal production and operation on index.\\
    Consequently, rules are in the following form: 
    \begin{tabular}{ll}
        Operation:&$(A, \mathrm{g}) \to (B, \mathrm{op})$\\
        AND:&$(A, \mathrm{g}) \to (BC, id)$\\
        OR:&$(A, \mathrm{g}) \to (B \mid C, id)$\\
        Dumping:&$(A, \mathrm{g}) \to (\bar{a}, id)$
    \end{tabular}\\
    where $\mid$ is a syntactic sugar for non-determinism.
    \\
    For every non-terminal, we allow at most one rule in the previous form, and for every rule in the previous form, we add a complementary guard rule. 
    \begin{align*}
        &\text{Complementary:}&&(A, \neg\mathrm{g}) \to \square
    \end{align*}
    The complementary rules guarantee that every indexed non-terminal generates a tree. The $\square$ symbol allows us to distinguish derivation  trees generated by these rules.
    \subsection{Computing transducer}\label{subsec:comp_trans}
        \cref{thm:gram_to_trans} is obtained by providing an explicit transducer $T_1(\gram)$.
        Its main idea is to simulate each derivation tree with an order $2$ stack, and to encode non-determinism via $push_1$ operations.
        We ensure the absence of deadlocks in the transducer by the guarantee that every indexed non-terminal generates a tree thanks to the absence of infinite derivations and the complementary rules. 
        The transformation of complementary rules ensures no $a$ is produced in the subrun corresponding to trees with $\square$. 
    \begin{definition}\label{def:transf_grammar_to_trans}
        Let $\gram = \bigl(N, T, \struc = <\mathcal{S}, \pred, \op>, P, S\bigr)$ be a $\struc$-grammar in normal form as defined in \cref{subsec:gram_normal_form}.\\
        We define \\
        $T_1(\gram):= \biggl(\mathcal{Q}_{T_1(\gram)}, \{a\}, \Delta_{T_1(\gram)} , q_0, \stack\Bigl(\{X_1,\#_1\}, \stack\bigl(\{N, \#_2\}, \struc\bigr)\Bigr), Start_{T_1(\gram)}, End_{T_1(\gram)}\biggr)$\\
        where :
        \begin{itemize}
            \item $\mathcal{Q}_{T_1(\gram)}=\{q_0\} \cup \{q_r, q_{r}' \mid r \in P\}$
            \item $\Delta_{T_1}= \bigl\{(q_0, \mathcal{P}_{X_1} \wedge \mathcal{P}_{\#_2}) \to (q_0, a, pop_1)\bigr\} \underset{r \in P}{\bigcup} T_1(r)$
         where $T_1(r)$ is defined by:\\
         \setlength{\tabcolsep}{4pt}
        \begin{tabular}{ll}
            If $r$ is a $operation$ rule then $T_1(r)=$&$
                                                        \left\{\begin{array}{l}
                                                            (q_0, \mathcal{P}_{X_1} \wedge \mathcal{P}_{A} \wedge \mathrm{g}) \to (q_r, \varepsilon, op)\\
                                                            (q_r, \top) \to \bigl(q_0, \varepsilon, push_2(B)\bigr)
                                                        \end{array}\right\}
                                                    $\\
            If $r$ is a $\mathrm{AND}$ rule then $T_1(r)=$&$\Bigl\{(q_0, \mathcal{P}_{X_1} \wedge \mathcal{P}_{A} \wedge \mathrm{g}) \to \bigl(q_0, \varepsilon, push_2(BC)\bigr)\Bigr\}$\\
            If $r$ is a $\mathrm{OR}$ rule then $T_1(r)=$&$
                                                                \left\{\begin{array}{l}
                                                                    (q_0, \mathcal{P}_{X_1} \wedge \mathcal{P}_{A} \wedge \mathrm{g}) \to \bigl(q_r, \varepsilon, push_2(C)\bigr)\\
                                                                    (q_r, \top) \to  \bigl(q_{r}', \varepsilon, push_1(X_1X_1)\bigr)\\
                                                                    (q_{r}', \top) \to \bigl(q_0, \varepsilon, push_2(B)\bigr)
                                                                \end{array}\right\}
                                                            $\\
            If $r$ is a $Dumping$ rule then $T_1(r)=$&$\bigl\{(q_0, \mathcal{P}_{X_1} \wedge \mathcal{P}_{A} \wedge \mathrm{g}) \to (q_0, \varepsilon, pop_2)\bigr\}$\\
            If $r$ is a $Complementary$ rule then $T_1(r)=$&$\bigl\{(q_0, \mathcal{P}_{X_1} \wedge \mathcal{P}_{A} \wedge \neg \mathrm{g}) \to (q_0, \varepsilon, pop_1)\bigr\}$
        \end{tabular}
            \item $Start_{T_1(\gram)}=\{X_1[S[\sigma]\#_2]\#_1 \mid \sigma \in \mathcal{S}\}$
            \item $End_{T_1(\gram)}=(q_0, \#_1)$
        \end{itemize}
    \end{definition}
    \begin{claimproof}[Proof Sketch of \cref{thm:gram_to_trans}]
        For the sake of simplicity, we prove this lemma for leftmost derivations (that corresponds to derivation trees; see section 4.3 in \cite{DBLP:books/aw/HopcroftU79}) instead of derivation trees.
        We prove  by structural induction that if there are $i$ leftmost derivation from a word $u\in (N\times \mathcal{S})^*$ in $\gram$ then the run from $(q, X_1[u\#_2]\Omega_1)$ produces the word $a^i$.\\
        The principal difficulty is for $\mathrm{OR}$ rules because of non-determinism. 
        This is solved by noticing that, in $\gram$, the number of leftmost derivation from $A[\sigma].u$  correspond to the addition of the number of leftmost derivation from $B[\sigma].u$ and  from $C[\sigma].u$  
        and that, in $T_1(\gram)$, the word produced from $(q_0, X_1[B_2[\sigma]u\#_2]X_1[C_2[\sigma]u\#_2]\Omega_1)$ is equivalent to the word 
        produced from $(q_0, X_1[B_2[\sigma]u\#_2]\Omega_1)$ concatenated with the word produced from $(q_0, X_1[C_2[\sigma]u\#_2]\Omega_1)$.\\
    \end{claimproof}
\section{From memory transducer to index grammar}\label{sec:trans_to_gram}
    For the rest of the paper, we fix $\aut=\Bigl(\mathcal{Q}_{\aut}, \{a\}, \Delta_\aut, q_0,\stack\bigl(\{X_1,\#_1\}, \stack(\Gamma, \struc)\bigr), Start_\aut,(q_{end}, \#_1)\Bigr)$ a $\stack^2(\struc)$-transducer 
    with $Start_\aut=\{X_1[Start_2[\sigma]\#_2]\#_1 \mid \sigma\in\mathcal{S}\}$, without cycle in the configuration graph, called after a livelock, from any configuration, 
    \emph{computing} a mapping $u: \struc \to \nat$ and  $\struc = <\mathcal{S}, \op, \mathcal{P}>$ a structure.\\
    In this part, we prove the following theorem:
    \begin{theorem}\label{thm:trans_to_gram}
    There exists a $\struc$-indexed grammar producing $u(\sigma)$ derivation trees from the axiom indexed by $\sigma$, for all $\sigma \in \mathcal{S}_\struc$. 
    \end{theorem}
    \begin{figure}
    \begin{tikzpicture}[shorten >= 1pt, node distance = 3.8cm, on grid, auto]
        \node [state, shape=ellipse, align=center] (q0) {$\aut$\\Computing\\transducer}; 
        \node (q0') [right =3cm of q0] {};
        \node [state, shape=ellipse, align=center] (q2) [below = 1.1cm of q0'] {$\aut_{\pred}$\\Transducer\\with predicates};
        \node [state, shape=ellipse, align=center] (q3) [right=of q2] {$\gram_{\pred}$\\Grammar\\with predicates};
        \node [state, shape=ellipse, align=center] (q4) [above=1.1cm of q0'] {$\aut_{\alt}$\\Alternating\\transducer};
        \node [state, shape=ellipse, align=center] (q6) [right=of q4] {$\gram_{\pr}$\\Proof\\grammar};
        \node (q6') [right =3cm of q6] {};
        \node [state, shape=ellipse, align=center] (q8) [below = 1.1cm of q6'] {$\gram$\\Enumerating\\grammar};
        \draw [->](q0) to [bend right] node [align=center] (note1) [midway, below left] {$T_2$\\\cref{subsubsec:pred_trans}}(q2);
        \draw [->](q2) to [bend right] node [align=center] (note1) [midway, below] {$T_3$\\\cref{subsubsec:pred_gram}}(q3);
        \draw [->](q0) to [bend left] node [align=center] (note1) [midway, above left] {$T_4$\\\cref{subsubsec:alt_trans}}(q4);
        \draw [->](q4) to [bend left] node [align=center] (note1) [midway, above] {$T_5$\\\cref{subsubsec:proof_gram}}(q6);
        \draw [->](q6) to [bend left] node [align=center] (note1) [midway, above right] {$T_6$\\\cref{subsec:enum_gramm}}(q8);
        \draw [->](q3) to [bend right] node [align=center] (note1) [midway, below right] {$T_6$\\\cref{subsec:enum_gramm}}(q8);
    \end{tikzpicture}
    \caption{Road map for \cref{thm:trans_to_gram}}\label{fig:roadmap}
    \end{figure}
    The lower part creates from $\struc$ a grammar enumerating the function $u$ whose trees correspond to productive $pop_1$ of the accepting run of $\struc$. 
    This grammar will use additional predicates to ensure its trees only represent correct portions of the accepting run of $\struc$. 
    The upper part creates a grammar in which non-terminals correspond to these additional  predicates :
    those non-terminals produces a single derivation  tree from a non-terminal if and only if the corresponding additional predicate is true and non otherwise. 
    Then, we merge the two grammars into one by replacing predicates with the corresponding non-terminals.\\
    The majority of the lemmas in this section are true for a $\stack^2(\struc)$-transducer with livelock from non-reachable configurations; 
    but the final merging step is not necessarily true for $\stack^2(\struc)$-transducer with livelock.
    \subsection{Transducer normal form}\label{subsec:trans_normal_form}
    Without loss of generality, we consider that, in $\aut$, every run from any configuration terminates with stack $\#_1$ \footnote{That we obtain by completing the transducer}, 
    there is only one order $1$ symbol in addition to the bottom-of-stack symbol, 
    operations $push_i$ push at most $2$ symbols and only rules using $pop_1$ can produce $a$.\\
    Consequently, there is a sink state $\bot$  and rules are in the following form:\\
    \begin{tabular}{ll}
        $renaming$:&$(q, \mathrm{g}) \to \bigl(p, \varepsilon, push_2(X_2)\bigr)$\\
        $push_2$:&$(q, \mathrm{g}) \to \bigl(p, \varepsilon, push_2(X_2 Y_2)\bigr)$\\
        $pop_1$:&$(q, \mathcal{P}_{\#_2}) \to \bigl(p, \bar{a}, pop_1\bigr)$
    \end{tabular}
    \begin{tabular}{ll}
        $operation$:&$(q, \mathrm{g}) \to (p, \varepsilon, op)$\\
        $pop_2$:&$(q, \mathrm{g}) \to (p, \varepsilon, pop_2)$\\
        $push_1$:&$(q, \mathrm{g}) \to \bigl(p, \varepsilon, push_1(X_1 X_1)\bigr)$
    \end{tabular}\\
    Rules on $\bot$ are $(\bot, \neg\mathcal{P}_{\#_2}) \to (\bot, \varepsilon, pop_2)$ and $(\bot, \mathcal{P}_{\#_2}) \to (p, \varepsilon, pop_1)$.\\
    Observe that this normal form preserves the absence of livelock.
    \subsection{Enumerating grammar with predicates }\label{subsec:pred_gram}
        \subsubsection{Transducer with predicates}\label{subsubsec:pred_trans}
            \begin{lemma}\label{lem:bij_trans_to_trans_pred}
                There exists a $\stack(\struc)$-transducer using additional predicates producing $u(\sigma)$ runs from $X[\sigma]\#$ for all $\sigma \in \mathcal{S}_\struc$. 
            \end{lemma}This lemma is obtained by providing a non-deterministic $\stack(\struc)$-transducer $T_2(\aut)$is such that runs from $(q, stk)$ correspond to $pop_1$ operations occurring in runs of $\aut$ from $(q, X_1[stk]\Omega)$.
            Thus, a run of $T_2(\aut)$ corresponding to a $pop_1$ operation simulates the order-$2$ stack that is popped by this operation, while keeping only the transitions that affects it.
            Runs whose corresponding $pop_1$ operations are productive are accepting; the others are rejecting. Hence, $T_2(\aut)$ has $i$ accepting runs from $(q, stk)$ if and only if the run of $\aut$ from $(q, X_1[stk]\Omega)$ produces $a^i$.\\
            \cref{fig:separation_of_order_1} represents, in black, a \emph{symbolic well-formed} run of $\aut$, and, in colors, the corresponding runs of $T_2(\aut)$.\\
            To achieve this, in $T_2(\aut)$, order-$2$ and $operation$ transitions of $\aut$ are preserved. Transitions producing $pop_1$ are replaced with transitions to an accepting state, while
            $push_1$ transitions are replaced by transitions that follow the new topmost $2$-stack, together with non-deterministic transitions that follow the remaining one below it (whose computation occurs after the corresponding $pop_1$).
            Consequently, we need to know the state in which the computation of $\aut$ will be after the $pop_1$ operation.This is the key technical difficulty in the proof of \cref{thm:trans_to_gram}.\\
            This is achieved by introducing $\mathcal{P}_{(p,r)}$, which is used in this section; its \emph{computability} will be established later.
            The determinism and completeness of $\aut$ ensure that, given a stack $stk$ and a state $p$, there exists exactly one state $r$ such that $\mathcal{P}_{(p,r)}$ holds.
            Thus, every $push_1$ rule of $\aut$ corresponds to exactly two rules of $T_2(\aut)$.\\
            The \emph{computing} property of $\aut$ ensures its determinism and totality. Therefore, no predicate is needed for the second stack added by $push_1$.\\
            \begin{figure}[h]
                \centering
                \begin{tikzpicture}
                    \draw[thick,->] (0,0) -- (11,0)node[anchor=south east] {};
                    \draw[thick,->] (0,0) -- (0,4)node[anchor=north west] {order 1 stack height};
                    \draw[step=1cm, gray, very thin] (0,0) grid (10.9,3.9);
                    \draw[thick,-] (0,1) -- (1,1);
                    \draw[thick,-] (1,1) -- (1.5,2);
                    \draw[thick,-] (1.5,2) -- (2,2);
                    \draw[thick,-] (2,2) -- (2.5,3);
                    \draw[thick,-] (2.5,3) -- (4,3);
                    \draw[thick,-] (4,3) -- (4.5,2);
                    \draw[thick,-] (4.5,2) -- (5,2);
                    \draw[thick,-] (5,2) -- (5.5,3);
                    \draw[thick,-] (5.5,3) -- (6,3);
                    \draw[thick,-] (6,3) -- (6.5,3);
                    \draw[thick,-] (6.5,3) -- (7,2);
                    \draw[thick,-] (7,2) -- (8,2);
                    \draw[thick,-] (8,2) -- (8.5,1);
                    \draw[thick,-] (8.5,1) -- (9.5,1);
                    \draw[thick,-] (9.5,1) -- (10,0);
                    \draw[thick, yellow,-] (0,0.95) -- (1,0.95);
                    \draw[thick, yellow,-] (1,0.95) -- (1.5,1.95);
                    \draw[thick, yellow,-] (1.5,1.95) -- (2,1.95);
                    \draw[thick, yellow,-] (2,1.95) -- (2.5,2.95);
                    \draw[thick, yellow,-] (2.5,2.95) -- (4,2.95);
                    \draw[thick, blue,-] (0,0.90) -- (1,0.90);
                    \draw[thick, blue,-] (1,0.90) -- (1.5,1.90);
                    \draw[thick, blue,-] (1.5,1.90) -- (2,1.90);
                    \draw[thick, blue, dashed] (2,1.9) -- (4.5,1.9);
                    \draw[thick, blue,-] (4.5,1.9) -- (5,1.9);
                    \draw[thick, blue,-] (5,1.90) -- (5.5,2.90);
                    \draw[thick, blue,-] (5.5,2.90) -- (6.5,2.90);
                    \draw[thick, green,-] (0,0.85) -- (1,0.85);
                    \draw[thick, green,-] (1,0.85) -- (1.5,1.85);
                    \draw[thick, green,-] (1.5,1.85) -- (2,1.85);
                    \draw[thick, green, dashed] (2,1.85) -- (4.5,1.85);
                    \draw[thick, green,-] (4.5,1.85) -- (5,1.85);
                    \draw[thick, green, dashed] (5,1.85) -- (7,1.85);
                    \draw[thick, green,-] (7,1.85) -- (8,1.85);
                    \draw[thick, red,-] (0,0.80) -- (1,0.80);
                    \draw[thick, red, dashed] (1,0.80) -- (8.5,0.80);
                    \draw[thick, red,-] (8.5,0.80) -- (9.5,0.80);
                    \draw [->](1,0.90) to [bend right,distance=0.8cm] node [align=center] (note1) [midway, below] {$\mathcal{P}_{\prescript{1}{}{q}, \prescript{1}{}{p}}$}(8.5,0.90);
                    \draw [->](2,1.90) to [bend right,distance=0.8cm] node [align=center] (note1) [midway, below] {$\mathcal{P}_{\prescript{2}{}{q}, \prescript{2}{}{p}}$}(4.5,1.90);
                    \draw [->](5,1.90) to [bend right,distance=0.8cm] node [align=center] (note1) [midway, below] {$\mathcal{P}_{\prescript{3}{}{q}, \prescript{3}{}{p}}$}(7,1.90);
                    \draw [-](0,-0) to [bend right] node [align=center] (note1) [midway, below] {$q_0$}(0,-0);
                    \draw [-](1,-0) to [bend right] node [align=center] (note1) [midway, below] {$\prescript{1}{}{q}$}(1,-0);
                    \draw [-](2,-0) to [bend right] node [align=center] (note1) [midway, below] {$\prescript{2}{}{q}$}(2,-0);
                    \draw [-](5,-0) to [bend right] node [align=center] (note1) [midway, below] {$\prescript{3}{}{q}$}(5,-0);
                    \draw [-](4.5,-0) to [bend right] node [align=center] (note1) [midway, below] {$\prescript{2}{}{p}$}(4.5,-0);
                    \draw [-](7,-0) to [bend right] node [align=center] (note1) [midway, below] {$\prescript{3}{}{p}$}(7,-0);
                    \draw [-](8.5,-0) to [bend right] node [align=center] (note1) [midway, below] {$\prescript{1}{}{p}$}(8.5,-0);
                    \draw [-](10,0) to [bend right] node [align=center] (note1) [midway, below] {$q_{fin}$}(10,0);
                \end{tikzpicture}
                \caption{visualization of \emph{symbolic well-formed} runs in a run }\label{fig:separation_of_order_1}
            \end{figure}
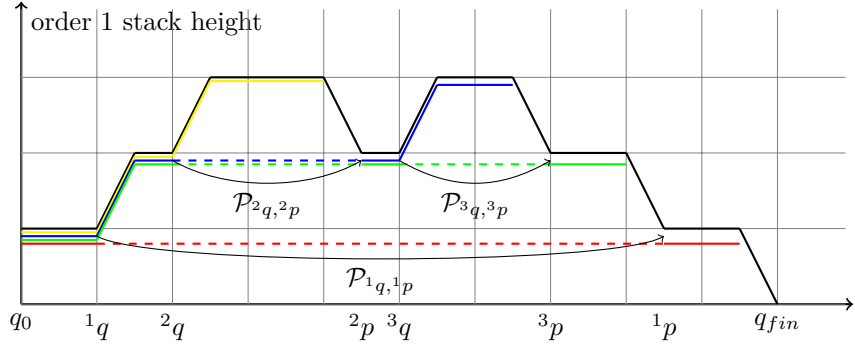
        \begin{definition}\label{def:transf_trans_to_transpred}
            We define $T_2(\aut):= \Bigl(\mathcal{Q}_\aut \cup \{\top\}, \emptyset, \Delta_{T_2(\aut)}, q_0, \stack_{T_2(\aut)}(\Gamma_2, \struc),Start_{T_2(\aut)}, (\top, \#)\Bigr)$ where $Start_{T_2(\aut)}=\{X[\sigma]\# \mid \sigma \in \mathcal{S}\}$ and  $\stack_{T_2(\aut)}$ is a stack augmented with predicates $\{\mathcal{P}_{q, p} \mid q, p \in \mathcal{Q}\}$ and where $\Delta_{T_2(\aut)} = \underset{\delta \in \Delta_\aut}{\bigcup} T_2(\delta)$ with:\\
            \begin{tabular}{lcl}
                If $\delta$ is an unproductive $pop_1$ rule then &$T_2(\delta)$ =&$\emptyset$\\
                If $\delta$ is a productive $pop_1$ rule then &$T_2(\delta)$ =&$\bigl\{(q, \mathcal{P}_{\#})\to (\top, \varepsilon, id)\bigr\}$\\
                If $\delta$ is a $push_1$ rule then &$T_2(\delta)$ =&$
                                                                \left\{\begin{array}{l}
                                                                    (q, \mathrm{g})\to(p, \varepsilon, id) \\
                                                                    (q, \mathrm{g}\wedge \mathcal{P}_{p, r})\to(r, \varepsilon, id)|r\in \mathcal{Q}
                                                                \end{array}\right\}
                                                                $\\
                Else &$T_2(\delta)$ =&$\{\delta\}$
            \end{tabular}\\
            Following the notation used in \cref{subsec:trans_normal_form}\\
            A predicate $\mathcal{P}_{p, r}(stk)$ is true if and only if $(p, A_1[stk]\Omega_1)\underset{\aut}{\xrightarrow{*}} (r, \Omega_1)$.
        \end{definition}
        \begin{claimproof}[Proof Sketch of \cref{lem:bij_trans_to_trans_pred}]
            The proof consists in a structural induction on \emph{symbolic well-formed} runs of $\aut$ from $q, X_1[stk]\Omega_1$. 
            It shows that for any such run with $k$ $pop_1$, there is $k$ runs of $T_2(\aut)$ from $q, stk$.
            The conclusion is obtained by the fact that the accepting run is a \emph{well-formed} run. 
        \end{claimproof}
        For the rest of the paper, we denote $T_2(\aut)$ as  $\aut_{\pred}=\Bigl(\mathcal{Q}_{\pred}, \emptyset, \Delta_{\pred}, q_0, \stack_{\pred} (\Gamma, \struc),$\\
        $\{X[\sigma]\#|\sigma\in \mathcal{S}_\struc\},(\top, \#)\Bigr)$.
        \subsubsection{Grammar with predicates}\label{subsubsec:pred_gram}
            \begin{lemma}\label{lem:trans_pred_to_gram_pred}
                There exists a $\struc$-grammar using additional predicates producing $u(\sigma)$ derivation trees from $S[\sigma]$ for all $\sigma \in \mathcal{S}_\struc$. 
            \end{lemma}This lemma is obtained by providing an explicit transformation $T_3$.
            The main idea is that each derivation tree rooted at the axiom indexed by $\sigma$ represents an accepting run of $\aut_{\pred}$ on $\sigma$.
            To achieve this, we use a \emph{triple grammar} (Chapter 5.4 in \cite{Mharrison}), in which a non-terminal is composed of the current state, the current symbol, and a guess of the state reached after popping the current atom.\\
            The predicates added by $\aut_{\pred}$ cannot be used directly, since they apply on a stack which is not explcit. 
            We therefore need to add a fourth component to the non-terminals.\\
            Consequently, each derivation tree rooted at an indexed non-terminal $(q, X, p, M)[\sigma]$ 
            represents a \emph{well-formed} run starting from the current atom $X[\sigma]$, in state $q$, and ending in state $p$.\\
            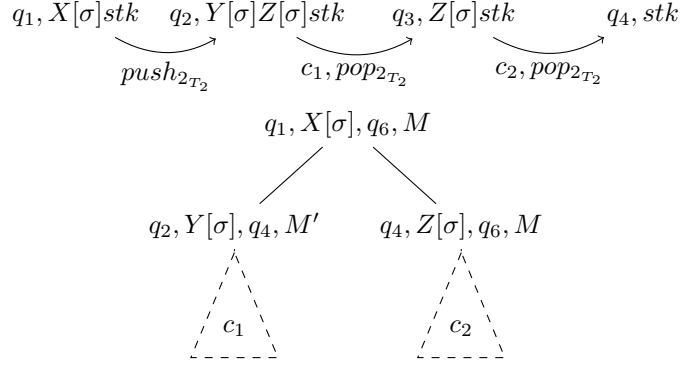
\begin{figure}[H]
                \centering
                \begin{tikzpicture}[shorten >= 1pt, node distance = 2cm, on grid, auto]
                    \node [align=center] (q0)  {$q_1, X[\sigma]stk$}; 
                    \node [align=center] (q1) [right=2.4cm of q0]  {$q_2, Y[\sigma]Z[\sigma]stk$}; 
                    \node [align=center] (q2) [right=2.6cm of q1]  {$q_3, Z[\sigma]stk$}; 
                    \node [align=center] (q3) [right=2.5cm of q2]  {$q_4, stk$}; 
                    \draw [->](q0) to [bend right] node [align=center]  [midway, below ] {$push_{2_{T_2}}$} (q1) ;
                    \draw [->](q1) to [bend right] node [align=center]  [midway, below ] {$c_1, pop_{2_{T_2}}$}(q2) ;
                    \draw [->](q2) to [bend right] node [align=center]  [midway, below ] {$c_2, pop_{2_{T_2}}$}(q3) ;
                \end{tikzpicture}
                \begin{tikzpicture}[level distance=1.35cm,
                        level 1/.style={sibling distance=3cm},
                        level 2/.style={sibling distance=4cm}]
                        \node {$q_1, X[\sigma], q_6, M$}
                            child {node {$q_2, Y[\sigma], q_4, M'$ }
                                child {node [draw, dashed, shape border uses incircle, isosceles triangle,
                                     shape border rotate=90, minimum height=10mm] {$c_1$}
                                }
                            }
                            child {node {$q_4, Z[\sigma], q_6, M$ }
                                child {node [draw, dashed, shape border uses incircle, isosceles triangle,
                                     shape border rotate=90, minimum height=10mm] {$c_2$}
                                }
                            }
                       ;
                    \end{tikzpicture}
                \caption{transformation of a \emph{well-formed} run of $\aut_{\pred}$ into derivation tree of $T_3(\aut_{\pred})$ }\label{fig:trans_pred_to_gram_pred}
            \end{figure}
            In a non-terminal, we memorize in $M$ all pairs of states $(q,p)$ such that $\mathcal{P}_{q,p}$ holds on the current simulated stack after popping the current element.
            More precisely, a non-terminal $(q, X, p, M)[\sigma]$ represents the equivalence class of configurations $(q, X[\sigma]stk)$ such that $M$ is the set of pairs of states $(q,p)$ for which $\mathcal{P}_{(q,p)}(stk)$ holds.
            The memories are computed step by step by adding predicates $\mathcal{P}_{M',M,Z}(\sigma)$ to rules simulating $push_{2_{T_2}}$. 
            This relies on the fact that $\mathcal{P}_{q,p}(A[\sigma]stk)$ holds if and only if there exists a state $r$ such that $\mathcal{P}_{r,p}(stk)$ holds and $(q, A[\sigma]stk) \xrightarrow{*}_{\aut_{\pred}} (r, stk)$.
            The computed memories are then used in rules simulating $push_1$, in place of the predicates $\mathcal{P}_{(q,p)}$.
            \begin{remark}\label{rem:totfunc}
            By determinism, completeness, and the absence of livelocks in $\aut$, we ensure that the possible memory values of non-terminals always form a total function.
            \end{remark}
            \begin{definition}\label{def:transf_transpred_to_pred_gram}
                We define $T_3(\aut_{\pred}):=(N, \emptyset, \struc_{T_3(\aut_{\pred})}, S, P_{T_3(\aut_{\pred})})$ with
                \begin{itemize}
                    \item $N=\bigl(\mathcal{Q}_{\pred}\times\Gamma\times \mathcal{Q}_{\pred}\times\mathcal{P} (\mathcal{Q}_{\pred}\times\mathcal{Q}_{\pred})\bigr)\cup \{S\}$ 
                    \item $\struc_{T_3(\aut_{\pred})}=\Bigl<\mathcal{S}, \mathcal{P}\cup\bigl\{\mathcal{P}_{M', M, X}|M', M\subseteq(\mathcal{Q}_{\pred}\times\mathcal{Q}_{\pred})\wedge X\in\Gamma\bigr\}, \op\Bigr>$
                    \item $P_{T_3(\aut_{\pred})}=$\begin{tabular}{ll}
                        &$\Biggl\{(S, \top)\to\Bigl(\bigl((q_0, X, p, M),(p, \#, \top, \emptyset)\bigr), id\Bigr)$\huge$\vert$\small$ p\in\mathcal{Q}_{\pred} $\\
                        &$\wedge M=\biggl\{(q, p)$\LARGE$\vert$\small $\begin{array}{l}
                                    \Bigl((q, \mathcal{P}_{\#_2})\to(p, a, pop_1)\in \Delta_\pred\Bigr)\vee\\ 
                                    \Bigl(p=\bot \wedge \forall r,\bigl((q, \mathcal{P}_{\#_2})\to(r, a, pop_1)\not\in\Delta_\pred\bigr)\Bigr)                     
                        \end{array}\biggr\}\Biggr\}  $\\
                        &$\cup \underset{\delta\in\Delta_{\aut_{\pred}}}{\bigcup}T_3(\delta)$ 
                    \end{tabular} 
                \end{itemize} 
                Where $T_3(\delta)$ is:\\ 
                \setlength{\tabcolsep}{1pt}
                \begin{tabular}{lcl}
                    If $\delta$ is $pop_{2_{T_2}}$ then &$T_3(\delta)$ = &$\Bigl\{\bigl((q, Z, p, M ), \mathrm{g}\bigr)\to(\varepsilon)|M\subseteq(\mathcal{Q}\times \mathcal{Q})\Bigr\}$\\
                    If $\delta$ is  $push_{2_{T_2}}$ then &$T_3(\delta)$ = &$\Bigl\{\bigl((q, Z, s, M), \mathrm{g}\wedge \mathcal{P}_{M', M, C}\bigr)\to\bigl((p, X, r, M')(r, Y, s, M), id\bigr)$\\
                    &&\Large$\vert$\small$(s\in\mathcal{Q}) \wedge  (p \in\mathcal{Q}) \wedge \bigl(M\subseteq(\mathcal{Q}\times\mathcal{Q})\bigr)\Bigr\}$\\
                    If $\delta$ is productive $pop_{1_{T_2}}$ then &$T_3(\delta)$ = &$\Bigl\{\bigl((q, \#, \top, \emptyset), \mathrm{g}\bigr)\to(\varepsilon)\Bigr\}$\\
                    If $\delta$ is second $push_{1_{T_2}}$ then &$T_3(\delta)$ =&$\Bigl\{\bigl((q, Z, s, M), \mathrm{g}\wedge\mathcal{P}_{M', M, X}\bigr)\to\bigl((r, X, s, M), id\bigr)$\\
                    &&\Large$\vert$\small$\bigl(s \in\mathcal{Q}\bigr)\wedge \bigl(M\subseteq(\mathcal{Q}\times\mathcal{Q})\bigr)\wedge \bigl(M\subseteq(\mathcal{Q}\times\mathcal{Q})\bigr) \wedge\bigl((p, r)\in M'\bigr)\Bigr\}$\\
                    Else &$T_3(\delta)$ =&$\Bigl\{\bigl((q, Z, s, M ), \mathrm{g}\bigr)\to\bigl((p, X, s, M), op\bigr)$\\
                    &&\Large$\vert$\small$\bigl(s \in \mathcal{Q}\bigr)\wedge \bigl(M\subseteq(\mathcal{Q}\times\mathcal{Q})\bigr)\Bigr\}$\\
                \end{tabular}\\
                Following the notation used in \cref{subsec:trans_normal_form} and \cref{def:transf_trans_to_transpred}\\
                With $Z$ being the order-2 symbol to which $\delta$ applies, and $M$ and $M'$ being functional and total sets of state pairs, i.e.
                $\forall q\in\mathcal{Q},\exists !, p\in\mathcal{Q}\ \text{such that}\ (q,p)\in M$.\\
                A predicate $\mathcal{P}_{M', M, Z} (\sigma)$ with $\sigma\in \mathcal{S}_\struc$ is true if and only if \\
                \begin{tabular}{l}
                $\forall stk\in \mathcal{S}_{\stack(\struc)}, \Bigl(M=\bigl\{(q, p)\mid q, X_1[stk]\Omega_1 \underset{\aut}{\xrightarrow{*}}p, \Omega_1 \bigr\}\Bigr)\implies$\\
                $\Bigl(M'=\bigl\{(q, p)| \exists r, p\in M, (q, X_1[Z_2[\sigma]stk]\Omega_1)\underset{\aut}{\xrightarrow{*}} (r, X_1[stk]\Omega_1)\bigr\}\Bigr)$.
                \end{tabular}
            \end{definition}
            \begin{claimproof}[Proof of \cref{lem:trans_pred_to_gram_pred}]
                We prove this lemma in two steps.\\
                We exhibit a total surjective function from the \emph{well-formed} runs of $\aut_{\pred}$ to the derivation trees of $T_3(\aut_{\pred})$ not rooted in  $S$. 
                This is done by promoting $T_3$ to a relation from the \emph{well-formed} runs of $\aut_{\pred}$ to the derivation trees of $T_3(\aut_{\pred})$ that are not rooted in  $S$, and by showing that this relation is a total function. 
                Surjectivity is achieved by establishing the totality of ${T_3}^{-1}$, both by structural induction. 
                The key idea is to notice that $T_3$ associates a run $(q, X[\sigma]stk) \underset{\aut_{\pred}}{\xrightarrow{*}} (p, stk)$ with the derivation tree rooted in  $(q, X, p, M)[\sigma]$, 
                where $M = \{(q, p) \mid q, X_1[stk]\Omega_1 \underset{\aut}{\xrightarrow{*}} (p, \Omega_1)\}$.\\                
                Then, we exhibit a bijection from the \emph{well-formed} runs of $\aut_{\pred}$ starting from $(q_0, X[\sigma]\#)$ to the derivation trees of $T_3(\aut_{\pred})$ rooted in  $S$. 
                We show that this relation is a total surjective function by proving that the rules on $S$ preserve the totality, the surjectivity and the functionality proved before. 
                Finally, we establish injectivity by showing that ${T_3}^{-1}$ is a function, since $stk$ is uniquely determined in every configuration appearing in an accepting run.
            \end{claimproof}
            For the rest of the paper, we denote $T_3(\aut_{\pred})$ as $\gram_{\pred}=(T_{\pred}, \emptyset, \struc_{\pred}, S, P_{\pred})$ 
            with $\struc_{\pred}=\Bigl<\mathcal{S}, \mathcal{P}\cup\bigl\{\mathcal{P}_{M', M, X}|M', M\subseteq(\mathcal{Q}_{\pred}\times\mathcal{Q}_{\pred})\wedge X\in\Gamma\bigr\}, op\Bigr>$.
    \subsection{Boolean output grammar}\label{subsec:bool_gram}
         We construct, from $\aut$, a $\struc$-grammar with the aim of replacing the predicates used in $T_3\bigl(T_2(\aut)\bigr)$.\
        This notion of grammar should not be confused with the notion of conjunctive grammar introduced by A. Okhotin \cite{DBLP:journals/jalc/Okhotin01}.
        Intuitively, our grammar returns a Boolean value in response to a query to the underlying transducer: the answer is true if and only if there exists a corresponding derivation tree, and false otherwise.
        This grammar respects the syntax given in the preliminaries.
        \subsubsection{Alternating transducer}\label{subsubsec:alt_trans}
        \begin{lemma}\label{lem:bij_trans_to_alt_trans}
            There exists a $\stack(\struc)$ alternating transducer having states in $(\mathcal{Q}\times\mathcal{Q})$ producing a unique accepting run tree from $\bigl((q, p), stk\bigr)$ if and only if there exists a \emph{symbolic well-formed} run from $(q, X_1[stk]\Omega_1)$ to $(p, \Omega_1)$ in $\aut$, and none otherwise.  
        \end{lemma}
        This lemma is obtained by providing an explicit transformation $T_4$.
        Its main idea is to contract runs into run trees whose leaves correspond to its $pop_1$. 
        In a run of $\aut$, rules are applied to the topmost order $2$ stack. 
        A branch of a run tree of $\aut_{\alt}$ simulates the part of the associated run of $\aut$ corresponding to an order $2$ stack or its predecessor. 
        For a run tree of $\aut_{\alt}$, the associated run of $\aut$ can be reconstructed by a depth-first search on the run tree. 
        When the topmost order $2$ stack changes via an order $1$ operation, we must ensure state coherence in $\aut_{\alt}$ despite the lack of state synchronization caused by the contraction. 
        Similary to \cref{subsubsec:pred_gram} ,this is achieved by assuming the state in which the topmost stack created by $push_1$ will be popped and by verifying this assumption in the rule obtained from $pop_1$. 
        If the guess is incorrect, the branch fails, hence also the full tree.
        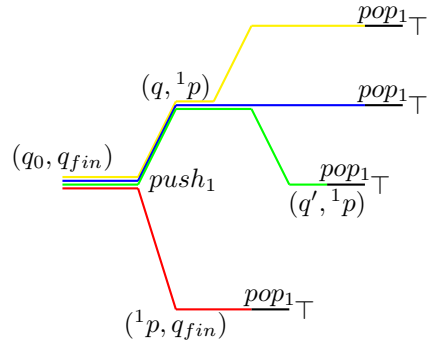
\begin{figure}[h]
                \centering
                \begin{tikzpicture}
                    \draw (0,1.2) node {$(q_0,q_{fin})$};
                    \draw[thick, yellow,-] (0,0.95) -- (1,0.95);
                    \draw (1.6,0.9) node {$push_1$};
                    \draw[thick, yellow,-] (1,0.95) -- (1.5,1.95);
                    \draw (1.5,2.15) node {$(q,\prescript{1}{}{p})$};
                    \draw[thick, yellow,-] (1.5,1.95) -- (2,1.95);
                    \draw[thick, yellow,-] (2,1.95) -- (2.5,2.95);
                    \draw[thick, yellow,-] (2.5,2.95) -- (4,2.95);
                    \draw[thick, -] (4,2.95) -- (4.5,2.95);
                    \draw (4.25,3.10) node {$pop_1$};
                    \draw (4.7,2.95) node {$\top$};
                    \draw[thick, blue,-] (0,0.90) -- (1,0.90);
                    \draw[thick, blue,-] (1,0.90) -- (1.5,1.90);
                    \draw[thick, blue,-] (1.5,1.90) -- (4,1.90);
                    \draw[thick,-] (4,1.90) -- (4.5,1.90);
                    \draw (4.25,2.05) node {$pop_1$};
                    \draw (4.7,1.9) node {$\top$};
                    \draw[thick, green,-] (0,0.85) -- (1,0.85);
                    \draw[thick, green,-] (1,0.85) -- (1.5,1.85);
                    \draw[thick, green,-] (1.5,1.85) -- (2.5,1.85);
                    \draw[thick, green,-] (2.5,1.85) -- (3,0.85);
                    \draw[thick, green,-] (3,0.85) -- (3.5,0.85);
                    \draw (3.5,0.6) node {$(q',\prescript{1}{}{p})$};
                    \draw (3.75,1) node {$pop_1$};
                    \draw[thick, -] (3.5,0.85) -- (4,0.85);
                    \draw (4.2,0.85) node {$\top$};
                    \draw[thick, red,-] (0,0.80) -- (1,0.80);
                    \draw[thick, red,-] (1,0.80) -- (1.5,-0.80);
                    \draw (1.5,-1.) node {$(\prescript{1}{}{p},q_{fin})$};
                    \draw[thick, red,-] (1.5,-0.80) -- (2.5,-0.80);
                    \draw[thick, -] (2.5,-0.80) -- (3,-0.80);
                    \draw (2.75,-0.65) node {$pop_1$};
                    \draw (3.2,-0.8) node {$\top$};
                \end{tikzpicture}
                \caption{visualization of the simulation of run in \cref{fig:separation_of_order_1} }\label{fig:run_tree}
            \end{figure}
        \begin{definition}\label{def:transf_trans_to_alt_trans}
            We define $T_4(\aut):= \bigl(\mathcal{Q}_{T_4(\aut)}, \emptyset, \Delta_{T_4(\aut)},q_{0_{T_4(\aut)}}, \stack (\Gamma_2, \struc),Start_{T_4(\aut)},(\top, \#)\bigr)$ where : 
            \begin{itemize}
                \item $\mathcal{Q}_{T_4(\aut)}=(\mathcal{Q}_\aut\times\mathcal{Q}_\aut)\cup\{\top\}$
                \item $\Delta_{T_4(\aut)}=\underset{\delta\in\Delta_\aut}{\bigcup}T_4(\delta)$ where $T_4(\delta)$ is defined by: \\
                    \begin{tabular}{lcl}
                        If $\delta$ is a $pop_1$ rule then &$T_4(\delta)$ =&$\Bigl\{\bigl((q, p), \mathcal{P}_{\#}\bigr)\to\bigl(\top, \varepsilon, id\bigr)\Bigr\}$\\
                        If $\delta$ is a $push_1$ rule then &$T_4(\delta)$ =&$\Bigl\{\bigl((q, s), \mathrm{g}\bigr)\to\bigl((p, r), \varepsilon, id\bigr)\wedge \bigl((r, s), \varepsilon, id\bigr)$\\
                        &&\Large$\vert$\small$(r\in\mathcal{Q})\wedge (s\in\mathcal{Q})\Bigr\}$\\
                        Else &$T_4(\delta)$ =&$\Bigl\{\bigl((q, s), \mathrm{g}\bigr)\to\bigl((p, s), \varepsilon, op\bigr)|s\in\mathcal{Q}\Bigr\}$
                    \end{tabular}\\
                    Following the notation used in \cref{subsec:trans_normal_form}.\\
                \item $q_{0_{T_4(\aut)}}=(q_0, q_{fin})$
                \item $Start_{T_4(\aut)}=\{Start_2[\sigma]\#|\sigma\in \mathcal{S}_\struc\}$
            \end{itemize}
            
        \end{definition}
        \begin{claimproof}[Proof Sketch of \cref{lem:bij_trans_to_alt_trans}]
            We exhibit a bijection between the \emph{symbolic well-formed} runs of $\aut$ and the accepting run trees of $T_4(\aut)$. 
            This is done by promoting $T_4$ to a relation from the \emph{symbolic well-formed} runs of $\aut$ to the accepting run trees of $T_4(\aut)$, and by showing that both $T_4$ and $T_4^{-1}$ are total functions.
            Both directions are solved by structural induction, using the determinism of $\aut$.
        \end{claimproof}
        We state a lemma that can be regarded as the key lemma of this part. 
        \begin{lemma}\label{lem:alt_trans_unic}(uniqueness property)
            For all $q\in\mathcal{Q}_\aut, X[\sigma]\in \Gamma_2\times\mathcal{S}_\struc$ and $\psi\in \mathcal{P} (\mathcal{Q}_\aut\times\mathcal{Q}_\aut)$ functional, 
            $\exists^{\leq1}p\in\mathcal{Q}_\aut, \varphi\subseteq\psi, \alpha\in\mathcal{T} (\Delta_{T_4(\aut)}), \bigl((q, p),X[\sigma]\Omega\bigr)\underset{T_4(\aut)}{\xrightarrow{\alpha}}\underset{(q_i, p_i)\in\varphi}{\bigwedge} \bigl((q_i, p_i), \Omega\bigr)$. 
        \end{lemma}
        From now on, we call $\phi_{q, X, \psi, \sigma}$ the subset of $\psi$ and $p_{q, X, \psi, \sigma}$ the state such that
        \begin{center}
            $\bigl((q, p_{q, X, \psi, \sigma}), X[\sigma].\Omega\bigr) \xrightarrow{*} \bigl(\underset{(q_i, p_i)\in\phi_{q, X, \psi, \sigma}}{\bigwedge} (q_i, p_i), \Omega\bigr)$
        \end{center}
        This uniqueness property is necessary in the definition of the enumerating grammar. In it, when eliminating the additional predicates, we need to know $\phi_{q, X, \psi, \sigma}$. 
        This can be handled with a non-determistic guess, as thanks to the uniqueness property, at most one guess can produce a tree.\\
        As a consequence of this lemma, there exists at most one $(q', p_{q, X, \psi, \sigma})$ ending state for the lower branch,
        i.e. the branch corresponding to the second symbol for every $push_{1_{T_4}}$, for a state $q$, an atom $X[\sigma]$, and a functional set $\psi$.
        We denote this state by $(q_{q, X, \psi, \sigma}, p_{q, X, \psi, \sigma})$.\\
        By definition,
        \begin{tabular}{l}
            $\forall X[\sigma]\in \Gamma_2\times\mathcal{S}_\struc, \forall q \in \mathcal{Q}, \forall stk\in \mathcal{S}_{\stack^2(\struc)}, \forall \psi \in \mathcal{P}(\mathcal{Q}\times\mathcal{Q}),$\\
            $ \bigl(\forall (p, r)\in\psi, (p, X_1[stk]\Omega_1)\underset{\aut}{\xrightarrow{*}}(r, \Omega_1)\bigr)$ \\
            $\implies \bigl((q, X_1[X_2[\sigma]stk]\Omega_1)\underset{\aut}{\xrightarrow{*}}(q_{q, X, \psi, \sigma}, X_1[stk]\Omega_1)\bigr)$
        \end{tabular}
        \begin{claimproof}[Proof Sketch of \cref{lem:alt_trans_unic}]
            For the sake of simplicity, we prove this lemma by proving it for any stack $stk$ instead of $X[\sigma]$.
            We establish this property by structural induction on run trees. By the determinism of $\aut$, only one rule $r$ can be applied to a given configuration in $\aut$. 
            Consequently, only rules produced by $T_4(r)$ can be applied in $T_4(\aut)$. 
            The only minor difficulty concerns the rules obtained from the transformation of $push_1$ rules due to the union of the two branches. 
            This is solved by the conservation of $p$ through the run tree in the upper branch, fixing the state $q$ in the lower branch. 
        \end{claimproof}
        For the rest of the paper, we denote $T_4(\aut)$ as $\aut_{\alt}=\bigl(\mathcal{Q}_{\alt}, \emptyset, \Delta_{\alt}, q_0, \stack(\Gamma_2, \struc), \{X[\sigma]\#|\sigma\in \mathcal{S}_\struc\},(\top, \#)\bigr)$.
        \subsubsection{Proof grammar}\label{subsubsec:proof_gram}
            \begin{lemma}\label{lem:bij_alt_trans_to_proof_gram}
                There exists a grammar producing one derivation  tree from $(X,\{(q, p)\}\to \varphi)[\sigma]$ with $\varphi$ functional if and only if 
                $\bigl((q, p), X[\sigma]\Omega\bigr) \xrightarrow[\aut_{\alt}]{*} \underset{(q'_i, p'_i)\in \varphi}{\wedge} \bigl((q'_i, p'_i), \Omega\bigr)$, and zero otherwise.
            \end{lemma}
            This lemma is obtained by providing an explicit transformation $T_5$.
            Its main idea is ,once again, to represent each atom of a stack $\aut_{\alt}$ as a non-terminal.
            Once again, the idea is similar to that of a triple grammar, but adapted to alternating transducers.
            Due to the fact that alternating transducers produce run trees rather than runs, the third component is the set of states reached after popping the current atom, rather than a single state.
            These sets are separated by rules simulating $push_1$ and are verified as singletons during the simulation of $pop$ rules.
            As an extension of \cref{lem:alt_trans_unic}, thanks to \cref{lem:bij_alt_trans_to_proof_gram}, we ensure that the grammar produces exactly one derivation tree from an indexed non-terminal $(X,{(q, p)}\to \varphi)[\sigma]$, and no derivation tree otherwise, if and only if $p=p_{q, X, \psi, \sigma}$ and $\varphi=\phi_{q, X, \psi, \sigma}$ for all $\psi$ such that $\varphi\subseteq\psi$.
            \begin{definition}\label{def:transf_alt_trans_to_proof_gram}
                We consider an arbitrary total order $<$ on the states of $\aut_{\alt}$.
                We define
                \begin{center}
                    $T_5(\aut_{\alt}):=\Bigl(\Gamma\times\bigl(\mathcal{P}(\mathcal{Q}_\aut\times\mathcal{Q}_\aut)\times\mathcal{P}(\mathcal{Q}_\aut\times\mathcal{Q}_\aut)\bigr), \emptyset, \struc, P=\bigl(\mathrm{Decomp} \cup\underset{r\in{\Delta_{\alt}}}{\bigcup}T_5(r)\bigr)\Bigr)$ 
                \end{center}
                a grammar-scheme where\\
                \begin{tabular}{lcl}
                    &$Decomp$=&$\Biggl\{\biggl(\Bigl(X,\bigl(\{(q, p)\}\cup\psi\bigr)\to(\varphi_1\cup\varphi_2)\Bigr), \top\biggr)$\\
                    &&$\to\biggl(\Bigl(A, (\{(q, p)\}\to \varphi_1)\Bigr)\Bigl(A, (\psi\to \varphi_2)\Bigr), id\biggr)$\\
                    &&\Huge$\vert$\small$\psi, \varphi_1, \varphi_2\in \mathcal{P}(\mathcal{Q}_{\alt}), \forall v\in\psi, (q, p)<v \Biggr\}$\\
                    If $r$ is $op_{T_4}$ then &$T_5(r)$ =&$\Biggl\{\biggl(\Bigl(Z,\bigl(\{(q, s)\}\to\varphi\bigr)\Bigr), g\biggr)$\\
                    &&$\to\biggl(\Bigl(Z,\bigl(\{(p, s)\}\to\varphi\bigr)\Bigr), op\biggr)$\Huge$\vert$\small$\varphi\in \mathcal{P}(\mathcal{Q}_{\alt})\Biggr\}$\\
                    If $r$ is $pop_{2_{T_4}}$ then &$T_5(r)$ =&$\Biggl\{\biggl(\Bigl(Z,\bigl(\{(q, s)\}\to\{(p, s)\}\bigr)\Bigr), g\biggr)\to\biggl(\varepsilon, id\biggr)\Biggr\}$\\
                    If $r$ is $renaming_{2_{T_4}}$ then &$T_5(r)$ =&$\Biggl\{\biggl(\Bigl(Z,\bigl(\{(q, s)\}\to\varphi\bigr)\Bigr), g\biggr)$\\
                    &&$\to\biggl(\Bigl(X, \bigl(\{(p, s)\}\to \varphi\bigr)\Bigr), id\biggr)$\Huge$\vert$\small$\varphi\in \mathcal{P}(\mathcal{Q}_{\alt})\Biggr\}$\\
                    If $r$ is  $push_{2_{T_4}}$ then &$T_5(r)$ =&$\Biggl\{\biggl(\Bigl(Z,\bigl(\{(q, s)\}\to\varphi\bigr)\Bigr), g\biggr)\to$\\
                    &&$\biggl(\Bigl(X, \bigl(\{(p, s)\}\to \psi\bigr)\Bigr)\Bigl(Y, \bigl(\psi\to \varphi\bigr)\Bigr), id\biggr)$\\
                    &&$|\psi, \varphi\subseteq\mathcal{P} (\mathcal{Q}_{\alt})\text{ functional }\Biggr\}$\\
                    If $r$ is $pop_{1_{T_4}}$ then &$T_5(r)$ =&$\Biggl\{\biggl(\Bigl(\#,\bigl(\{(q, p)\}\to\top\bigr)\Bigr), \top\biggr)\to\biggl(\varepsilon, id\biggr)\Biggr\}$\\
                    If $r$ is $push_{1_{T_4}}$ then &$T_5(r)$ =&$\Biggl\{\biggl(\Bigl(Z,\bigl(\{(q, s)\}\to(\varphi_1\cup\varphi_2)\bigr)\Bigr), g\biggr)$\\
                    &&$\to\biggl(\Bigl(Z,\bigl(\{(p, r)\}\to \varphi_1\bigr)\Bigr)\Bigl(Z,\bigl(\{(r, s)\}\to \varphi_2\bigr)\Bigr), id\biggr)$\\
                    &&\Huge$\vert$\small$\varphi_1, \varphi_2\in \mathcal{P}(\mathcal{Q}_{\alt})\Biggr\}$\\
                \end{tabular}\\           
                Following the notation used in \cref{subsec:trans_normal_form} and \cref{def:transf_trans_to_alt_trans}
            \end{definition}
            \begin{claimproof}[Proof Sketch of \cref{lem:bij_alt_trans_to_proof_gram}]
                We consider the promotion of $T_5$ as a relation from the \emph{symbolic well-formed} run trees to the derivation  trees.
                We prove that $T_5^{-1}$ is a total function using \cref{lem:alt_trans_unic}.
                Then, we prove a uniqueness property lemma, that follows from \cref{lem:alt_trans_unic}, for the derivation trees of the grammar scheme and we use it to prove that $T_5$ is a total function.
                
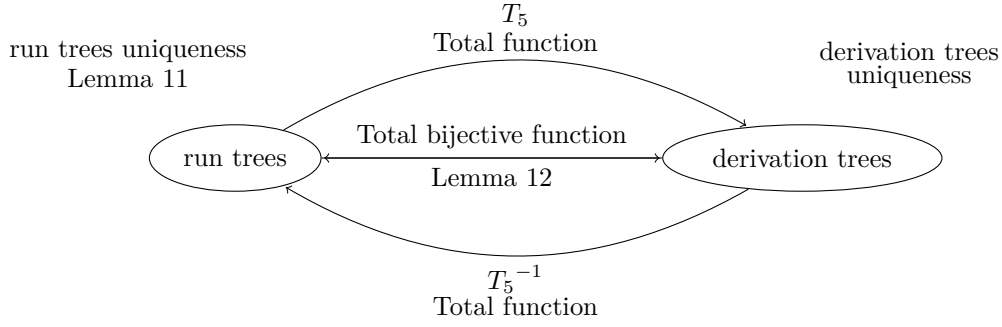
\begin{center}
                    \begin{tikzpicture}[shorten >= 1pt, node distance = 2cm, on grid, auto]
                        \node [state, shape=ellipse] (q0)  {run trees}; 
                        \node [state, ellipse] (q1) [right=7.5cm of q0] {derivation trees};  
                        \draw [->](q0) to [bend left] node (note1) [midway, above] {Total function} (q1) ;
                        \node (note1') [above=10pt of note1] {$T_5$}; 
                        \draw [->](q1) to [bend left] node (note2) [midway, below] {${T_5}^{-1}$} (q0) ;
                        \node (note2') [below=10pt of note2] {Total function}; 
                        \node (note2'') [below=10pt of note2'] {}; 
                        \node (1) [above right=of q1] {derivation trees}; 
                        \node (1') [below=10pt of 1] {uniqueness}; 
                        \node (2) [above left=of q0] {run trees uniqueness }; 
                        \node (2') [below=10pt of 2] { \cref{lem:alt_trans_unic} };
                        \draw [->](q0) to  node (0) [midway, above] {Total bijective function}  (q1) ;
                        \draw [->](q1) to  node (0') [midway, below] {\cref{lem:bij_alt_trans_to_proof_gram}} (q0);
                    \end{tikzpicture}
                \end{center}
                In all parts of this proof, we consider a promotion $T_5$ as a relation from the \emph{symbolic well-formed} run trees of $\aut_{\alt}$ to the derivation  trees of $T_5(\aut_{\alt})$ rooted in $\bigl(X,(q, p)\to\varphi\bigr)[\sigma]$.
                The relation $T_5$ associates a \emph{symbolic well-formed} run tree $\bigl((q, p), X[\sigma]\Omega\bigr) \xrightarrow{c} \bigl(\varphi, \Omega\bigr)$ with a derivation tree from $\bigl(X,(q, p)\to\varphi\bigr)[\sigma]$.\\
                We first prove, via structural induction from $\bigl(X,(q, p)\to\varphi\bigr)[\sigma]$ with $\varphi$ functional, that $T_5^{-1}$ is total. 
                Functionality follows immediately, since there is at most one \emph{symbolic well-formed} run tree from $\bigl((q, p), X[\sigma]\Omega\bigr)$ to $\bigl(\varphi, \Omega\bigr)$ with $\varphi$ functional, as established in \cref{lem:alt_trans_unic}.\\
                Secondly, we prove via structural induction that $T_5$ is total. 
                We prove that $T_5$ is functional if there is a uniqueness property for derivation trees of $T_5(\aut_{\alt})$ rooted in  $\bigl(X,(q, p)\to\varphi\bigr)$ with $\varphi$ functional.\\
                So, finally, we prove that for all non-terminals $\bigl(X,(q, p)\to\varphi\bigr)$ with $\varphi$ functional, there exists at most one derivation tree rooted in the non-terminal.
                There is a minor difficulty for  $push_{T_5}$ due to the addition of $\psi$. This is solved by the order $<$ on states of $\aut_{\alt}$, ensuring uniqueness in the separation made by $Decomp$, and the functionality of $\psi$ allowing the use of \cref{lem:alt_trans_unic}.
            \end{claimproof}
            For the rest of the paper, we denote $T_5(\aut_{\alt})$ as $\gram_{\pr}=(T_{\pr}, \emptyset, \struc, P_{\pr})$.
        \subsection{Enumerating grammar}\label{subsec:enum_gramm}
            We give a transformation creating, from $\gram_{\pred}$ and $\gram_{\pr}$, a $\struc$-grammar satisfying \cref{thm:trans_to_gram}.
            This transformation consists in a substitution of each predicate $\mathcal{P}_{M', M, A}$ in $\gram_{\pred}$ by a concatenation of non-terminals $\bigl(A,(q_i, p_i)\to M_i\bigr)$ of $\gram_\pr$ with $(q_i, p_i)\in M'$ and $M_i\subseteq M$. 
            Thanks to the uniqueness property in $\gram_{\pr}$ for every subset of a functional set, and $M$ is functional by \cref{rem:totfunc}, 
            we can use non-determinism to guess the correct $M_i$ for every element of $M'$ without multiplying the number of derivation trees. 
            \begin{definition}\label{def:transf_proof_gram_to_enum_gram} 
                We define $T_6(\gram_{\pred}, \gram_{\pr}):=\bigl(T_{\pred} \cup T_{\pr}, \emptyset, \struc,$\\$ P_{\pr} \cup \underset{r\in P_{\pred}}{\bigcup} T_6(r), S_{\pred}\bigr)$.\\
                With $T_6$:
                \setlength{\tabcolsep}{1pt}
                \begin{tabular}{lllclllll}
                    If&$r $&$=$&&$ \Bigl(A, g$&$\wedge \mathcal{P}_{M', M, Y}$&$\Bigr)\to \Bigl(\gamma$&&$, op\Bigr)$ \\
                    then& $T_6(r)$&$=$&$ \biggl\{$&$\Bigl(A, g$&&$\Bigr)\to \Bigl(\gamma$&$\underset{0\leq i\leq \#(\mathcal{Q}_\aut)}{\odot}\bigl(Y,(q_i, p_i)\to M_{i}\bigr)$&$, op\Bigr)$\\
                    &&&&&\huge$\vert$\small$ \underset{0\leq i\leq \#(\mathcal{Q}_\aut)}{\forall}$&$\Bigl(\bigl(q_i, p_i$&$\bigr)\in M' \wedge M_{i}\subseteq M\Bigr)\biggr\}$\\
                    Else&$T_6(r)$&$=$&$\{r\}$\\
                \end{tabular}  
        \end{definition}
        \begin{proof}[Proof of \cref{thm:trans_to_gram}]
            By \cref{lem:trans_pred_to_gram_pred}, the grammar $T_6(\gram_{\pred}, \gram_{\pr})$ produces $u(\sigma)$ derivations trees from $S_{\pred}$ for all $\sigma\in \mathcal{S}_\struc$ if the substitution of $\mathcal{P}_{M', M, X}$  produces one derivation tree if and only if the predicate is true and zero otherwise.\\
            By \cref{lem:alt_trans_unic} and \cref{lem:bij_alt_trans_to_proof_gram}, we know that there exists a derivation tree from $\bigl(X,(q_i, p_i)\to M_i\bigr)[\sigma]$ only if $(q_{q_i, X, M, \sigma}, p_{q_i, X, M, \sigma})\in M_i$.
            By rewriting the property given for $q_{q_i, X, M, \sigma}$ in \cref{subsubsec:alt_trans}, we have that $\forall stk\in \mathcal{S}_{\stack(\struc)},$ such that $M=\bigl\{(q, p)\mid q, X_1[stk]\Omega_1 \underset{\aut}{\xrightarrow{*}}p, \Omega_1\bigr\}$  \\
            \begin{tabular}{ll}
                $\Bigl(\bigl(X,(q_i, p_i)\to M_i\bigr)[\sigma]$ produces a derivation tree &$\implies$\\
                $\underbrace{\exists (q'_i, p_i) \in M, \bigl(q_i, X_1[X_2[\sigma]stk]\Omega_1\bigr) \underset{\aut}{\xrightarrow{*}}\bigl(q'_i, X_1[stk]\Omega_1\bigr)\Bigr)}_{(\text{denoted } P)}$&\\
            \end{tabular}\\
            As said in \cref{subsubsec:proof_gram}, $(X,(q_i, p_i)\to M_i)[\sigma]$ can produce a derivation tree if and only if $M_i=\phi_{q_i, X, M, \sigma}$. Consequently, \\
            \begin{tabular}{ll}    
                $\forall stk\in \mathcal{S}_{\stack(\struc)}, \Bigl(M=\bigl\{(q, p)\mid q, X_1[stk]\Omega_1 \underset{\aut}{\xrightarrow{*}}p, \Omega_1 \bigr\}\Bigr)\implies$\\
                $\bigl(\underset{M_i\subseteq M}{\oplus}\bigl(X,(q_i, p_i)\to M_i\bigr)[\sigma]$  produce a derivation tree $\Longleftrightarrow P\bigr)$\\
            \end{tabular}\\
            With $\oplus$ a notation for non-determinism.\\
            As $\aut$ is complete and without livelock, we know that all incorrect non-terminal will die in a finite derivation tree.\\
            So, \\
            \begin{tabular}{ll}    
                $\forall stk\in \mathcal{S}_{\stack(\struc)}, \Bigl(M=\bigl\{(q, p)\mid q, X_1[stk]\Omega_1 \underset{\aut}{\xrightarrow{*}}p, \Omega_1 \bigr\}\Bigr)\implies$\\
                $\Bigl(\underset{0\leq i\leq \#(\mathcal{Q}_\aut)}{\odot}\underset{M_i\subseteq M}{\oplus}\bigl(X,(q_i, p_i)\to M_i\bigr)[\sigma]$ produces a derivation tree $\Longleftrightarrow \forall (q_i, p_i)\in M', P\Bigr)$.
            \end{tabular}\\
            Which is equivalent to the definition of $\mathcal{P}_{M', M, X}[\sigma]$ given in \cref{def:transf_transpred_to_pred_gram}, thanks to the totality of  $M$ and $M'$ as stated in \cref{rem:totfunc}.
        \end{proof}
\section{Conclusions and further work}
Using effective proofs, we have shown that functions \emph{computed} by $\stack^2(\struc)$-transducers are equivalent to functions enumerating trees generated by $\struc$-grammars.
Since a stack built over a structure is itself a structure, we obtain the following corollary:
\begin{corollary}
The following classes of functions are equivalent:
\begin{itemize}
\item functions computed by $\stack^{k+2}(\struc)$-transducers;
\item functions enumerating trees generated by $\stack^k(\struc)$-grammars.
\end{itemize}
In the case where $\struc$ is the structure on $\nat$ endowed with the operation $n \mapsto n\stackrel{.}{-}1$ and the non-zero predicate,
since the sequences defined by polynomial recurrences are computed by deterministic transducers with memory in stacks of stacks of $\nat$ (\cite{DBLP:journals/apal/FrataniS06}, Proposition 53), 
our result shows that such sequences can be represented by context-free grammars with indexes in $\nat$.\\
This link between transducers and combinatorial objects, can be thought of as an extension of the well-known links between $\nat$-linear recurrences and regular grammars, $\nat$-algebraic recurrences and context-free grammars (see the surveys \cite{Viennot85,DBLP:conf/dimacs/Delest94,DBLP:conf/stacs/Bousquet-Melou05}).\\  
Since they are contained in series defined by polynomial recurrences, this framework could be used to identify syntactic restrictions on indexed grammars to answer the representation problem for D-finite series.
\end{corollary}


\newpage
\bibliography{papier_gram_trans}
\newpage
\appendix
\tableofcontents  
\section{Annexe}\label{sec:annexe}
\subsection{\cref{subsec:gram_normal_form} proof}

    In this part we will proof the following lemma.

    \begin{lemma}\label{lem:gram_normal_form}
        There exists a transformation from a $\struc$-grammar to a normal form that preserves the number of $\square$-free derivation trees generated from the axiom indexed by $\sigma \in \mathcal{S}_\aut$.
    \end{lemma}

    We prove this with two transformation. First, we show a transformation such that the resulting grammar satisfy the required rule form. 
    Second, we show a transformation such that the resulting grammar satisfy the “only one rule per non-terminal” condition using only rule in the required form. 
    Finally, we add the complementary rules.

\paragraph{First transformation}
Let $\gram=(N, T, \struc, P, S)$ be a $\struc$-grammar. \\
We considere that every non-terminal have at least one rule applying on them. If a non-terminal $A$ infringes these properties we solve it by adding the rule $(A, \bot)\to(\varepsilon, id)$   
We first define a transformation $T_{sep}$ from a $\struc$-grammar to a grammar satisfying the required rule forms.

\begin{definition}   
We define 
\[
T_{sep}(\gram)=(N, T_{T_{sep}(\gram)}=T\cup \underset{\underset{\to(\beta, op_r)\in P}{r=(p_r, g_r)}}{\bigcup}\{A_{r_i}, B_{r_i}|0<i\leq \#(\beta)\}, \struc, P_{T_{sep}(\gram)}= \underset{r\in P}{\bigcup}T(r), S).
\]
Let $r=(p_r, g_r)\to(\alpha_{r_1}\alpha_{r_2}\dots\alpha_{r_k}, op_r)$, where $\alpha_i\in(N\cup T)$. Then:
\begin{align*}
T(r)=&\{(p_r, g_r)\to(A_{r_1}, op_r)\}\\
&\cup\{(A_{r_i}, \top)\to(B_{r_i}A_{r_{i+1}}, id)|i<k\}\\
&\cup\{(B_{r_i}, \top)\to(\alpha_i, id)|i\leq k\}\\
&\cup\{(A_{r_k}, \top)\to(B_{r_k}, id)\}
\end{align*}
\end{definition}

According to its definition, $T_{sep}(\gram)$ satisfies the required rule forms.

\begin{lemma}
Let $\gram=(N, T, \struc, P, S)$ be a $\struc$-grammar that produces $u(\sigma)$ derivation trees from $S[\sigma]$. Then the $\struc$-grammar $T_{sep}(\gram)$ produces $u(\sigma)$ derivation trees rooted in  $S[\sigma]$.
\end{lemma}

\begin{proof}
Let $\gram=(N, T, \struc, P, S)$ be a $\struc$-grammar, and let 
$T_{sep}(\gram)=(N, T_{T_{sep}(\gram)}, \struc, P_{T_{sep}(\gram)}, S)$.

We exhibit a bijection between derivation trees of $\gram$ and derivation trees of $T_{sep}(\gram)$ rooted in  non-terminals in $T$. This is done by promoting $T_{sep}$ to a relation between derivation trees of $\gram$ and derivation trees of $T_{sep}(\gram)$ rooted in  non-terminals in $T$, and showing that $T_{sep}$ and ${T_{sep}}^{-1}$ are total functions.

\proofsubparagraph*{Direction $1$}
We first prove that the promotion of $T_{sep}$ from derivation trees of $\gram$ to derivation trees of $T_{sep}(\gram)$ rooted in  non-terminals in $T$ is a total function. The proof proceeds by structural induction on derivation trees of $\gram$.

Let $d$ be a derivation tree whose first applied rule is 
$r=(p_r, g_r)\to(\alpha_{r_1}\alpha_{r_2}\dots\alpha_{r_k}, op_r)$, 
where $\alpha_i \in (N\cup T)$:
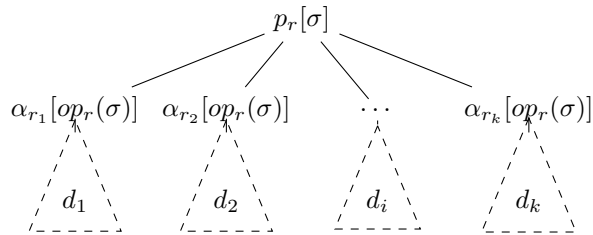
\begin{center}
\begin{tikzpicture}[level distance=1.2cm,
    level 1/.style={sibling distance=2cm},
    level 2/.style={sibling distance=2cm}]
    \node {$p_r[\sigma]$}
        child {
            node {$\alpha_{r_1}[op_r(\sigma)]$}
                child {node [draw, dashed, shape border uses incircle, isosceles triangle,
                shape border rotate=90, minimum height=10mm] {$d_1$}}
        }
        child {
            node {$\alpha_{r_2}[op_r(\sigma)]$} 
                child {node [draw, dashed, shape border uses incircle, isosceles triangle,
                shape border rotate=90, minimum height=10mm] {$d_2$}}
        }
        child {
            node {$\dots$} 
                child {node [draw, dashed, shape border uses incircle, isosceles triangle,
                shape border rotate=90, minimum height=10mm] {$d_i$}}
        }
        child {
            node {$\alpha_{r_k}[op_r(\sigma)]$}
                child {node [draw, dashed, shape border uses incircle, isosceles triangle,
                shape border rotate=90, minimum height=10mm] {$d_k$}}
        };
\end{tikzpicture}
\end{center}

Then $T_{sep}(d)$ is:

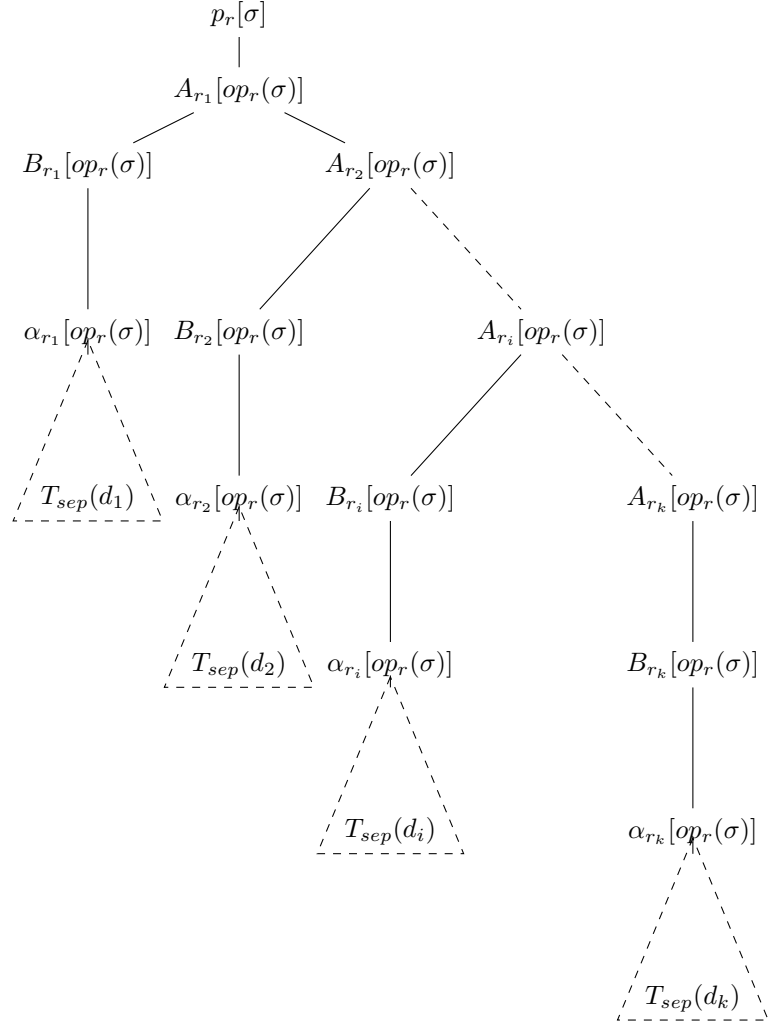
\begin{center}
\begin{tikzpicture}[level distance=1cm,
    level 1/.style={sibling distance=1cm},
    level 2/.style={sibling distance=4cm},
    level 3/.style={sibling distance=4cm, level distance=2.2cm}]
    \node {$p_r[\sigma]$}
        child {
            node {$A_{r_1}[op_r(\sigma)]$}
            child{
                node{$B_{r_1}[op_r(\sigma)]$}
                child{
                    node{$\alpha_{r_1}[op_r(\sigma)]$}
                    child {
                        node [draw, dashed, isosceles triangle,
                        shape border rotate=90, minimum height=10mm] {$T_{sep}(d_1)$}
                    }
                }
            }
            child{
                node{$A_{r_2}[op_r(\sigma)]$}
                child{
                    node{$B_{r_2}[op_r(\sigma)]$}
                    child{
                        node{$\alpha_{r_2}[op_r(\sigma)]$}
                        child {
                            node [draw, dashed, isosceles triangle,
                            shape border rotate=90, minimum height=10mm] {$T_{sep}(d_2)$}
                        }
                    }
                }
                child[dashed]{
                    node{$A_{r_i}[op_r(\sigma)]$}
                    child[solid]{
                        node{$B_{r_i}[op_r(\sigma)]$}
                        child{
                            node{$\alpha_{r_i}[op_r(\sigma)]$}
                            child {
                                node [draw, dashed, isosceles triangle,
                                shape border rotate=90, minimum height=10mm] {$T_{sep}(d_i)$}
                            }
                        }
                    }
                    child[dashed]{
                        node{$A_{r_k}[op_r(\sigma)]$}
                        child[solid]{
                            node{$B_{r_k}[op_r(\sigma)]$}
                            child{
                                node{$\alpha_{r_k}[op_r(\sigma)]$}
                                child {
                                    node [draw, dashed, isosceles triangle,
                                    shape border rotate=90, minimum height=10mm] {$T_{sep}(d_k)$}
                                }
                            }
                        }
                    }
                }
            }
        };
\end{tikzpicture}
\end{center}

By definition, there exists exactly one rule associated with each non-terminal $A_{r_i}$ and $B_{r_i}$. By the induction hypothesis, $T_{sep}(d_i)$ is uniquely defined for each $i$. Hence, $T_{sep}(d)$ is uniquely determined. Therefore, the promotion of $T_{sep}$ defines a total function on derivation trees.

\proofsubparagraph*{Direction $2$}
We now prove that $T_{sep}^{-1}$, from derivation trees of $T_{sep}(\gram)$ rooted in  non-terminals in $T$ to derivation trees of $\gram$, is a total function. The proof proceeds by structural induction on derivation trees of $T_{sep}(\gram)$ rooted in  non-terminals in $T$.

The only rules on non-terminals in $T$ are of the form $(p_r, g_r)\to(A_{r_1}, op_r)$. 
Let $t$ be a derivation tree of $T_{sep}(\gram)$ rooted in  $p_r$, whose first applied rule is $(p_r, g_r)\to(A_{r_1}, op_r)$. By definition, this rule is followed by rules on $A_{r_i}$ and $B_{r_i}$.

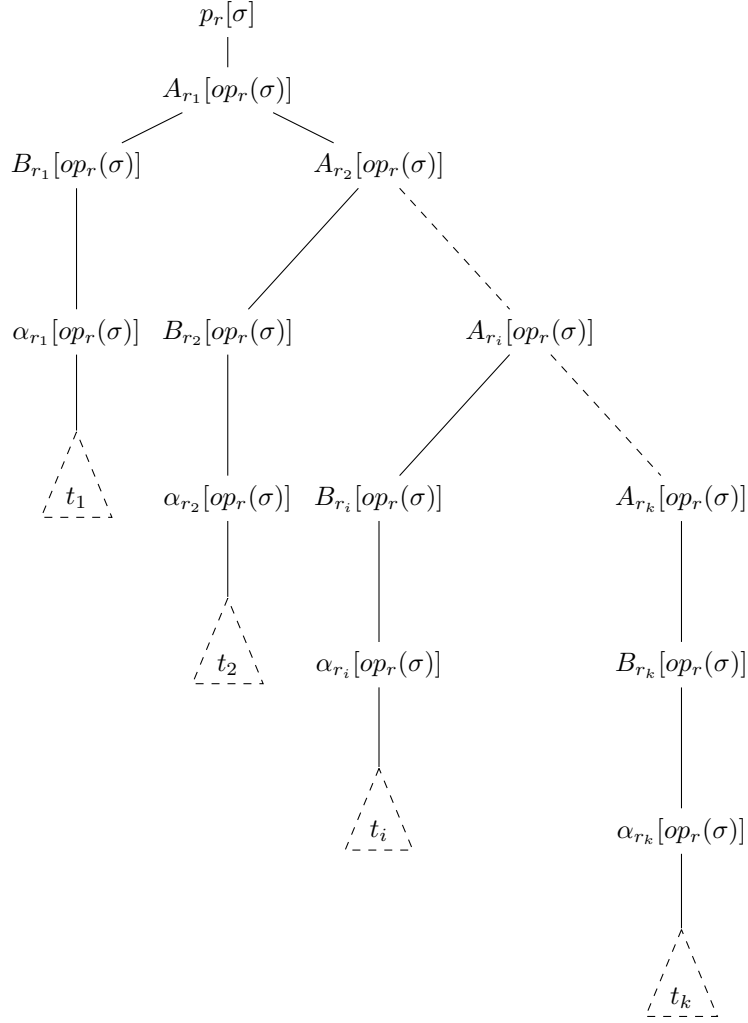
\begin{center}
\begin{tikzpicture}[level distance=1cm,
    level 1/.style={sibling distance=1cm},
    level 2/.style={sibling distance=4cm},
    level 3/.style={sibling distance=4cm, level distance=2.2cm}]
    \node {$p_r[\sigma]$}
        child {
            node {$A_{r_1}[op_r(\sigma)]$}
            child{
                node{$B_{r_1}[op_r(\sigma)]$}
                child{
                    node{$\alpha_{r_1}[op_r(\sigma)]$}
                    child {
                        node [draw, dashed, isosceles triangle,
                        shape border rotate=90, minimum height=10mm] {$t_1$}
                    }
                }
            }
            child{
                node{$A_{r_2}[op_r(\sigma)]$}
                child{
                    node{$B_{r_2}[op_r(\sigma)]$}
                    child{
                        node{$\alpha_{r_2}[op_r(\sigma)]$}
                        child {
                            node [draw, dashed, isosceles triangle,
                            shape border rotate=90, minimum height=10mm] {$t_2$}
                        }
                    }
                }
                child[dashed]{
                    node{$A_{r_i}[op_r(\sigma)]$}
                    child[solid]{
                        node{$B_{r_i}[op_r(\sigma)]$}
                        child{
                            node{$\alpha_{r_i}[op_r(\sigma)]$}
                            child {
                                node [draw, dashed, isosceles triangle,
                                shape border rotate=90, minimum height=10mm] {$t_i$}
                            }
                        }
                    }
                    child[dashed]{
                        node{$A_{r_k}[op_r(\sigma)]$}
                        child[solid]{
                            node{$B_{r_k}[op_r(\sigma)]$}
                            child{
                                node{$\alpha_{r_k}[op_r(\sigma)]$}
                                child {
                                    node [draw, dashed, isosceles triangle,
                                    shape border rotate=90, minimum height=10mm] {$t_k$}
                                }
                            }
                        }
                    }
                }
            }
        };
\end{tikzpicture}
\end{center}

By construction, the rules $(p_r, g_r)\to(A_{r_1}, op_r)$ together with rules on $A_{r_i}$ and $B_{r_i}$ exists if and only if there exists a rule 
$r=(p_r, g_r)\to(\alpha_{r_1}\dots\alpha_{r_k}, op_r)$ in $\gram$.

By the induction hypothesis, ${T_{sep}}^{-1}(t_i)$ exists in $\gram$ and is uniquely determined for each $i$. Hence, ${T_{sep}}^{-1}(t)$ is uniquely determined. Therefore, $T_{sep}^{-1}$ defines a total function.
\end{proof}

\paragraph{Second transformation}
We next give a transformation $T_{1rule}$ from a $\struc$-grammar obtained by $T_{sep}$ to a grammar satisfying the "only one rule per non-terminal" property.

\begin{definition}   
We define 
\[
T_{1rule}(\gram)=(N_{T_{1rule}(T_{sep}(\gram))}, T, P_{T_{1rule}(T_{sep}(\gram))}=\bigcup_{A\in T_{T_{sep}(\gram)}}\bigoplus_{r_i \text{ on } A} r_i, S).
\]
For two rules $r_1=(A, g_{r_1})\to(\beta_{r_1}, op_{r_1})$ and $r_2=(A, g_{r_2})\to(\beta_{r_2}, op_{r_2})$ applying to the same non-terminal, we define:
\begin{align*}
r_1 \bigoplus r_2 = &\{(A, \top)\to(C_{r_1}C_{r_2}, id), \\
&(C_{r_1}, g_{r_1})\to(\beta_{r_1}, op_{r_1}), \\
&(C_{r_2}, g_{r_2})\to(\beta_{r_2}, op_{r_2})\}
\end{align*}

For a set $R$ of rules $r_i$ applying to the same non-terminal, we define:
\begin{align*}
&\bigoplus_{r_i \in R} r_i = r_i && \text{if } \#(R)=1, \\
&\bigoplus_{r_i \in R} r_i = (\ldots (r_1 \bigoplus r_2)\bigoplus \ldots )\bigoplus r_{\#(R)} && \text{otherwise.}
\end{align*}
\end{definition}

\begin{lemma}
Let $\gram=(N, T, \struc, P, S)$ be a $\struc$-grammar that produces $u(\sigma)$ derivation trees from $S[\sigma]$. Then the $\struc$-grammar $T_{1rule}(T_{sep}(\gram))$ produces $u(\sigma)$ derivation trees rooted in  $S[\sigma]$.
\end{lemma}

\begin{proof}
Let $\gram_{sep}=T_{sep}(\gram)$ and $T_{1rule}(\gram_{sep})=(N, T_{T_{1rule}(\gram_{sep})}, \struc, P_{T_{1rule}(\gram_{sep})}, S)$.

We exhibit a bijection between derivation trees of $\gram_{sep}$ and derivation trees of $T_{1rule}(\gram_{sep})$ rooted in  indexed non-terminals in $T_{\gram_{sep}}$. This is done by promoting $T_{1rule}$ to a relation between derivation trees of $\gram_{sep}$ and derivation trees of $T_{1rule}(\gram_{sep})$ rooted in  indexed non-terminals in $T_{\gram_{sep}}$, and showing that $T_{1rule}$ and ${T_{1rule}}^{-1}$ are total functions.

\proofsubparagraph*{Direction $1$}
We first prove that the promotion of $T_{1rule}$ from derivation trees of $\gram_{sep}$ to derivation trees of $T_{1rule}(\gram_{sep})$ rooted in  indexed non-terminals in $T_{\gram_{sep}}$ is a total function. The proof proceeds by structural induction on derivation trees of $\gram_{sep}$. 

By definition, all non-terminals added by the transformation $T_{sep}$ have only one rule. Therefore, by the definition of $\bigoplus$, $T_{1rule}$ does not transform rules on these non-terminals, and the induction is trivial for derivation  trees rooted in  these non-terminals.

Let $A$ be a non-terminal in $T$. Let $R=\{r_1, \ldots, r_{\#(R)}\}$ be the set of rules on $A$ and $D=\{d_1, \ldots, d_{\#(R)}\}$ the corresponding derivation  tree of $\gram_{sep}$ rooted in  $A$, such that the first applied rule in $d_i$ is $r_i=(A, g_i)\to(A_i, op_i)$.
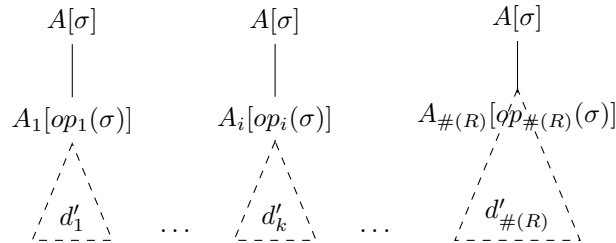
\begin{center}
\begin{tikzpicture}[level distance=1.3cm,
    level 1/.style={sibling distance=1cm},
    level 2/.style={sibling distance=4cm},
    level 3/.style={sibling distance=4cm, level distance=2.2cm}]
    \node {$A[\sigma]$}
        child {
            node {$A_{1}[op_{1}(\sigma)]$}
            child{
                node [draw, dashed, isosceles triangle,
                shape border rotate=90, minimum height=10mm] {$d'_1$}
            }
        }
       ;
\end{tikzpicture}
\begin{tikzpicture}[level distance=1.3cm,
    level 1/.style={sibling distance=1cm},
    level 2/.style={sibling distance=4cm},
    level 3/.style={sibling distance=4cm, level distance=2.2cm}]
    \node {$\ldots $}
       ;
\end{tikzpicture}
\begin{tikzpicture}[level distance=1.3cm,
    level 1/.style={sibling distance=1cm},
    level 2/.style={sibling distance=4cm},
    level 3/.style={sibling distance=4cm, level distance=2.2cm}]
    \node {$A[\sigma]$}
        child {
            node {$A_{i}[op_{i}(\sigma)]$}
            child{
                node [draw, dashed, isosceles triangle,
                shape border rotate=90, minimum height=10mm] {$d'_k$}
            }
        }
       ;
\end{tikzpicture}
\begin{tikzpicture}[level distance=1.3cm,
    level 1/.style={sibling distance=1cm},
    level 2/.style={sibling distance=4cm},
    level 3/.style={sibling distance=4cm, level distance=2.2cm}]
    \node {$\ldots $}
       ;
\end{tikzpicture}
\begin{tikzpicture}[level distance=1.3cm,
    level 1/.style={sibling distance=1cm},
    level 2/.style={sibling distance=4cm},
    level 3/.style={sibling distance=4cm, level distance=2.2cm}]
    \node {$A[\sigma]$}
        child {
            node {$A_{\#(R)}[op_{\#(R)}(\sigma)]$}
            child{
                node [draw, dashed, isosceles triangle,
                shape border rotate=90, minimum height=10mm] {$d'_{\#(R)}$}
            }
        }
       ;
\end{tikzpicture}
\end{center}
Then $T_{1rule}(d_i)$ is:  
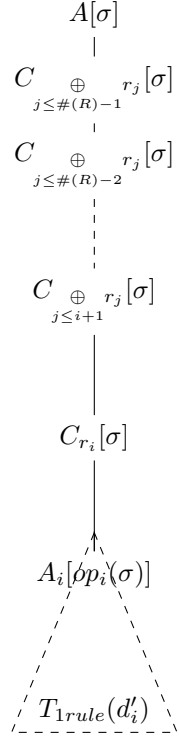
\begin{center}
\begin{tikzpicture}[level distance=1cm,
    level 1/.style={sibling distance=1cm},
    level 2/.style={sibling distance=4cm},
    level 3/.style={sibling distance=4cm, level distance=1.8cm}]
    \node {$A[\sigma]$}
        child {node {$C_{\underset{j\leq\#(R)-1}{\oplus}r_j}[\sigma]$ }
            child {node {$C_{\underset{j\leq\#(R)-2}{\oplus}r_j}[\sigma]$ }
                child {node {$C_{\underset{j\leq i+1}{\oplus}r_j}[\sigma]$ }[dashed]
                    child{ node {$C_{r_i}[\sigma]$} [solid]
                        child{node  {$A_{i}[op_{i}(\sigma)]$} 
                            child{node [draw, dashed, isosceles triangle,
                            shape border rotate=90, minimum height=10mm] {$T_{1rule}(d'_i)$}
                            }
                        }
                    }
                }
            }
        }
   ;
\end{tikzpicture}
\end{center}By the induction hypothesis, $T_{1rule}(d'_i)$ is uniquely determined. Hence $T_{1rule}(d_i)$ is uniquely determined. Therefore, the promotion of $T_{1rule}$ defines a total function on derivation trees.

\proofsubparagraph*{Direction $2$}
We now prove that $T_{1rule}^{-1}$, from derivation trees of $T_{1rule}(\gram_{sep})$ rooted in  indexed non-terminals in $T_{\gram_{sep}}$ to derivation trees of $\gram_{sep}$, is a total function. The proof proceeds by structural induction.

By definition, all non-terminals added by $T_{sep}$ have only one rule. Therefore, by definition of $\bigoplus$, the induction is trivial for derivation  trees rooted in  these non-terminals.

The only rules on non-terminals in $T$ are of the form $(A, \top)\to(C_{\bigoplus_{j\leq k-1}r_j}|C_{r_k}, id)$. Let $t$ be a derivation  tree of $T_{1rule}(\gram_{sep})$ rooted in  $A$:

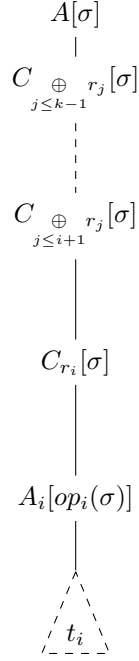
\begin{center}
\begin{tikzpicture}[level distance=1cm,
    level 1/.style={sibling distance=1cm},
    level 2/.style={sibling distance=4cm, level distance=1.8cm}]
    \node {$A[\sigma]$}
        child {node {$C_{\underset{j\leq k-1}{\oplus}r_j}[\sigma]$ }
            child {node {$C_{\underset{j\leq i+1}{\oplus}r_j}[\sigma]$ }[dashed]
                child{ node {$C_{r_i}[\sigma]$} [solid]
                    child{node  {$A_{i}[op_{i}(\sigma)]$} 
                        child{node [draw, dashed, isosceles triangle,
                        shape border rotate=90, minimum height=10mm] {$t_i$}
                        }
                    }
                }
            }
        }
   ;
\end{tikzpicture}
\end{center}
By definition, such a rule exists if and only if there are $k$ rules on $A$ in $\gram_{sep}$. By induction, there exists exactly one derivation  tree $d$ such that $T_{1rule}(d)=t_i$. 

Therefore, $T_{1rules}^{-1}$ is a total function. 
    \end{proof}
\subsection{\cref{thm:gram_to_trans} proof}
We prove that the word produce by the well formed run from $q_0, X_1[u\#_2]$ in $T_1(\gram)$, that we write $t_{T_1(\gram)}(X_1[u\#_2])$, is equal to the number of leftmost derivation from $u$ without $\square$ on leaves in $\gram$ , that we write $\#\lm_\gram(u)$, for every $u\in (N\times \mathcal{S})^*$, i.e. $$\forall u\in(N\times \mathcal{S})^*, t_{T_1(\gram)}(X_1[u\#_2])=\#LmD_\gram(u)$$\\

Since $\gram$ is in normal form, only one rule $r$ can be applied to the leftmost letter of $u$.\\

\begin{description}
    \item [If $r$ is a $Complementary$ rule] $(A, \neg g)\to ( \square, id)$ then $\#LmD_\gram(u)=0$.\\ 
        By $T_1$, there is a rule $(q_0, \mathcal{P}_{X_1}\wedge \mathcal{P}_{A_2} \wedge \neg g)\to (q_0, \varepsilon, pop_1)$. \\
        By normal form and determinism of $T_1$, there is no other rules on $(q_0, X_1[u\#_2])$, so $t_{T_1(\gram)}(X_1[u\#_2])=t_{T_1(\gram)}(\varepsilon)=0$.\\
        Then  $t_{T_1(\gram)}(X_1[u\#_2])=\#LmD_\gram(u)=0$.
    \item[If $r$ is a $Dumping$ rule] $(A, g)\to ( \bar{a}, id)$ then $\#LmD_\gram(u)=\#LmD_\gram(u')=i$.\\ 
        By $T_1$, there is a rule $(q_0, \mathcal{P}_{X_1}\wedge \mathcal{P}_{A_2} \wedge g)\to (q_0, \varepsilon, pop_2)$. \\
        By induction $t_{T_1(\gram)}(X_1[u'\#_2])=i$. \\
        By normal form and determinism of $T_1$, there is no other rules on $(q_0, X_1[u\#_2])$, so $t_{T_1(\gram)}(X_1[u\#_2])=t_{T_1(\gram)}(X_1[u'\#_2])=i$.\\
        Then  $t_{T_1(\gram)}(X_1[u\#_2])=\#LmD_\gram(u)=i$.
    \item [If $r$ is a $Operation$ rule] $(A, g)\to ( B, op)$ then $\#LmD_\gram(u)=\#LmD_\gram(u')=i$.\\ 
        By $T_1$, there is rules 
        \begin{tabular}{ll}
            $\delta_1$: &$(q_0, \mathcal{P}_{X_1}\wedge \mathcal{P}_{A_2} \wedge g)\to (q_r, \varepsilon, op)$\\
            $\delta_2$: &$(q_r, \top)\to (q_0, \varepsilon, push_2(B_2))$\\
        \end{tabular}\\
        By induction $t_{T_1(\gram)}(X_1[u'\#_2])=i$. \\
        By $T_1$ there is no other rule than $\delta_2$ on $q_r$.\\
        By normal form and determinism of $T_1$, there is no other rules on $(q_0, X_1[u\#_2])$, so $t_{T_1(\gram)}(X_1[u\#_2])=t_{T_1(\gram)}(X_1[u'\#_2])=i$.\\
        Then  $t_{T_1(\gram)}(X_1[u\#_2])=\#LmD_\gram(u)=i$.
    \item [If $r$ is a $\mathrm{AND}$ rule $(A, g)\to ( BC, id)$] then $\#LmD_\gram(u)=\#LmD_\gram(u')=i$.\\ 
        By $T_1$, there is a rule $(q_0, \mathcal{P}_{X_1}\wedge \mathcal{P}_{A_2} \wedge g)\to (q_0, \varepsilon, push_2(B_2C_2))$. \\
        By induction $t_{T_1(\gram)}(X_1[u'\#_2])=i$. \\
        By normal form and determinism of $T_1$, there is no other rules on $(q_0, X_1[u\#_2])$, so $t_{T_1(\gram)}(X_1[u\#_2])=t_{T_1(\gram)}(X_1[u'\#_2])=i$.\\
        Then  $t_{T_1(\gram)}(X_1[u\#_2])=\#LmD_\gram(u)=i$.
    \item[If $r$ is a $\mathrm{OR}$ rule $(A, g)\to ( B|C, id)$] then $\#LmD_\gram (A[\sigma]u)=\#LmD_\gram (B[\sigma]u) + \#LmD_\gram(C[\sigma]u)=i+j$.\\ 
        By $T_1$, there is rules 
        \begin{tabular}{ll}
            $\delta_1$: &$(q_0, \mathcal{P}_{X_1}\wedge \mathcal{P}_{A_2} \wedge g)\to (q_r, \varepsilon, push_2(C_2))$\\
            $\delta_2$: &$(q_r, \top)\to (q_{r}', \varepsilon, push_2(B_2))$\\
            $\delta_3$: &$(q_{r}', \top)\to (q_0, \varepsilon, push_2(B_2))$\\
        \end{tabular}\\
        By induction $t_{T_1(\gram)}(X_1[B_2[\sigma]u\#_2])=i$ and $t_{T_1(\gram)}(X_1[C_2[\sigma]u\#_2])=j$. \\
        By $T_1$ there is no other rule than $\delta_2$ on $q_r$ and no other rule than $\delta_3$ on $q'_r$.\\
        By normal form and determinism of $T_1$, there is no other rules on $(q_0, X_1[u\#_2])$, so $t_{T_1(\gram)}(X_1[A_2[\sigma]u\#_2])=t_{T_1(\gram)}(X_1[B_2[\sigma]u\#_2]X_1[C_2[\sigma]u\#_2])$.\\
        Notice that $t_{T_1(\gram)}(X_1[B_2[\sigma]u\#_2]X_1[C_2[\sigma]u\#_2])=t_{T_1(\gram)}(X_1[B_2[\sigma]u\#_2])+t_{T_1(\gram)}(X_1[C_2[\sigma]u\#_2])$.\\
        So, $t_{T_1(\gram)}(X_1[A_2[\sigma]u\#_2])=t_{T_1(\gram)}(X_1[B_2[\sigma]u\#_2])+t_{T_1(\gram)}(X_1[C_2[\sigma]u\#_2])=i+j$.\\
        Then  $t_{T_1(\gram)}(X_1[A_2[\sigma]u\#_2])=\#LmD_\gram(A_2[\sigma]u)=i+j$.
\end{description}

If $A[\sigma]\to \varepsilon$ then $\#LmD_\gram(A[\sigma])=1$.\\
By rule $(q_0, \mathcal{P}_{X_1}\wedge \mathcal{P}_{\#_2})\to (q_0, a, pop_1)$, and the fact that it is the only rule on $\#_2$, then $t_{T_1(\gram)}(X_1[\#_2])=1$.\\
By induction $t_{T_1(\gram)}(X_1[A[\sigma]\#_2])=1$.\\
\subsection{proof of \cref{subsec:trans_normal_form}}
    In this part, we will prove the following lemma.
        \begin{lemma}\label{lem:trans_norm_form}
            There exists a function that transforms a $P^2(\aut)$-transducer into normal form while preserving the computed function.
        \end{lemma}
        We prove this result through three rule transformations.
        First, we show a transformation such that $push_i$ operations are binary and unary .
        Second, we show a tranformation such that only rules applying the $pop_1$ operation produce $a$ using only binary and unary $push_i$ operations .  
        Finally, we show a  transformation satisfying the “single order $1$ symbol” condition keeping the previous properties.

Let 
$
\aut=(\mathcal{Q}, T, \Delta,
\stack(\Gamma_1, \stack(\Gamma_2, \struc)), Z, End)
$
be a \emph{computing} $\stack^2(\struc)$-transducer.

The property allowing $pop_1$ only when the topmost order $2$ symbol is $\#_2$ can easily be ensured by adding a draining state performing $pop_2$ transitions until $\#_2$ is reached.  
Hence, we assume that this property holds in $\aut$.

\paragraph{First transformation}
As stated in the proof sketch, we first define a transformation $T_{bin}$ that decomposes every $n$-ary $push_i$, for $n>2$, into a sequence of unary and binary $push_i$ operations.

\begin{definition}
We define
$
T_{bin}(\aut)=
(\mathcal{Q}_{\aut_{bin}}, T,
\Delta_{\aut_{bin}},
\struc, Z, End)
$
where
$
\Delta_{\aut_{bin}}
=\bigcup_{r\in\Delta}T_{bin}(r).
$

For a rule $r=(q, g_r)\to(p, \bar a, op_r)$:

\begin{align*}
&\text{If } op_r\in \op_{\struc}\cup\{pop_i\},
&& T_{bin}(r)=\{r\}.\\
\\
&\text{If } op_r=push_i(A_{r_1}, \dots, A_{r_k}),
&& T_{bin}(r)= \\
&&& \{(q, g_r)\to(q_{r_1}, \bar a,
push_i(A_{r_{k-1}}, A_{r_k}))\}\\
&&&\cup\{(q_{r_j}, \top)\to(q_{r_{j+1}}, \varepsilon,
push_i(A_{r_{k-(j+1)}}, A_{r_{k-j}}))
\mid 1\le j<k-1\}\\
&&&\cup\{(q_{r_{k-1}}, \top)\to(p, \varepsilon, id)\}.
\end{align*}
\end{definition}
\begin{lemma}
Let $\aut$ be a $\stack^2(\struc)$-transducer \emph{computing} $u(\sigma)$. 
Then $T_{bin}(\aut)$ computes $u(\sigma)$.
\end{lemma}

\begin{proof}

We establish a bijection between symbolic well-formed runs of $\aut$ and symbolic well-formed runs of $T_{bin}(\aut)$ starting from stacks containing only symbols in $\Gamma_1$ and $\Gamma_2$.

We promote $T_{bin}$ as a relation between such runs and prove that both $T_{bin}$ and $T_{bin}^{-1}$ are total functions.

\medskip
\noindent
\proofsubparagraph*{Direction $1$}

We prove by structural induction on symbolic well-formed runs of $\aut$ that the promotion of $T_{bin}$ is a total function.

\medskip
\noindent
\emph{Case 1: $op_r\in\op_{\struc}\cup\{pop_i\}$.}

Let
\[
c=q, A_1[stk_2]\Omega_1
\xrightarrow{r}
\xrightarrow{c_1}
q', \Omega_1
\]
be a symbolic well-formed run of $\aut$.

Since $T_{bin}(r)=r$, by construction,
\[
T_{bin}(c)=
q, A_1[stk_2]\Omega_1
\xrightarrow{r}
\xrightarrow{T_{bin}(c_1)}
q', \Omega_1.
\]

By induction hypothesis, $T_{bin}(c_1)$ is uniquely defined. Hence so is $T_{bin}(c)$.

\medskip
\noindent
\emph{Case 2: $push_2(A_{r_1}, \dots, A_{r_k})$.}

Let
\begin{align*}
c&=q, A_1[A[\sigma]stk_2]\Omega_1\\
&\xrightarrow{r}
q, A_1[A_{r_1}[\sigma]\dots A_{r_k}[\sigma]stk_2]\Omega_1\\
&\xrightarrow{c_1}
q', \Omega_1.
\end{align*}

By construction,

\begin{align*}
T_{bin}(c)
&=(q, A_1[A[\sigma]stk_2]\Omega_1)\\
&\xrightarrow{push_2(A_{r_{k-1}}, A_{r_k})}
(q_{r_1}, A_1[A_{r_{k-1}}[\sigma]A_{r_k}[\sigma]stk_2]\Omega_1)\\
&\xrightarrow{*}
(q_{r_{k-1}}, A_1[A_{r_1}[\sigma]\dots A_{r_k}[\sigma]stk_2]\Omega_1)\\
&\xrightarrow{id}
(p, A_1[A_{r_1}[\sigma]\dots A_{r_k}[\sigma]stk_2]\Omega_1)\\
&\xrightarrow{T_{bin}(c_1)}
(q', \Omega_1).
\end{align*}

The uniqueness of intermediate states ensures determinism.  
The induction hypothesis concludes.

\medskip
\noindent
\emph{Case 3: $push_1(A_{r_1}, \dots, A_{r_k})$.}

Let

\begin{align*}
c&=q, A_1[stk_2]\Omega_1\\
&\xrightarrow{r}
p, A_{r_1}[stk_2]\dots A_{r_k}[stk_2]\Omega_1\\
&\xrightarrow{c_1}
\dots
\xrightarrow{c_k}
q', \Omega_1.
\end{align*}

Then

\begin{align*}
T_{bin}(c)
&=q, A_1[stk_2]\Omega_1\\
&\xrightarrow{push_1(A_{r_{k-1}}, A_{r_k})}
(q_{r_1}, A_{r_{k-1}}[stk_2]A_{r_k}[stk_2]\Omega_1)\\
&\xrightarrow{*}
(q_{r_{k-1}}, A_{r_1}[stk_2]\dots A_{r_k}[stk_2]\Omega_1)\\
&\xrightarrow{id}
(p, A_{r_1}[stk_2]\dots A_{r_k}[stk_2]\Omega_1)\\
&\xrightarrow{T_{bin}(c_1)}
\dots
\xrightarrow{T_{bin}(c_k)}
(q', \Omega_1).
\end{align*}

Uniqueness follows from construction and induction hypothesis.

     \proofsubparagraph*{Direction 2}
We prove that the relation $T_{bin}^{-1}$, from the symbolic well-formed runs of $T_{bin}(\aut)$ over stacks using only symbols in $\Gamma_1$ and $\Gamma_2$, to the symbolic well-formed runs of $\aut$, is a  total function.  
The proof proceeds by structural induction on symbolic well-formed runs of $T_{bin}(\aut)$ using only symbols in $\Gamma_1$ and $\Gamma_2$.

Let $c=q, A_1[stk_2]\Omega_1 \xrightarrow{r} \xrightarrow{c_1} r, \Omega_1$ be a symbolic well-formed run of $T_{bin}(\aut)$ starting with a rule $r=(q, g_r)\to(p, \bar{a}, op_r)$ such that $op_r\in \op_\struc\cup \{pop_i\}$.  

By construction, $T_{bin}^{-1}(r)=r$. Therefore,
\[
T_{bin}^{-1}(c) = q, A_1[stk_2]\Omega_1 \xrightarrow{r} \xrightarrow{T_{bin}^{-1}(c_1)} r, \Omega_1.
\]  
By induction, there exists exactly one $c'_1$, a symbolic well-formed run of $\aut$, such that $T_{bin}^{-1}(c_1)=c'_1$.  
Hence, there exists exactly one $c'$, a symbolic well-formed run of $\aut$, such that $T_{bin}^{-1}(c)=c'$.

Let $c$ be a symbolic well-formed run of $T_{bin}(\aut)$ starting with a rule $r=(q, g_r)\to(p, \bar{a}, push_2(A_{r_{k-1}}A_{r_k}))$.  
By construction, using the uniqueness of rules on added states:

\begin{align*}
&c = && (q, A_1[A[\sigma]stk_2]\Omega_1) \\
&\xrightarrow{push_2(A_{r_{k-1}}A_{r_k})} && (q_{r_1}, A_1[A_{r_{k-1}}[\sigma] A_{r_k}[\sigma] stk_2]\Omega_1) \\
&\xrightarrow{*} && (q_{r_i}, A_1[A_{r_i}[\sigma] A_{r_{i-1}}[\sigma] \ldots  A_{r_k}[\sigma] stk_2]\Omega_1) \\
&\xrightarrow{*} && (q_{r_{k-1}}, A_1[A_{r_1}[\sigma] A_{r_2}[\sigma] \ldots  A_{r_k}[\sigma] stk_2]\Omega_1) \\
&\xrightarrow{id} && (p, A_1[A_1[\sigma] \ldots  A_{r_k}[\sigma] stk_2]\Omega_1) \\
&\xrightarrow{c_1} && (r, \Omega_1)
\end{align*}

By construction, $r$ exists in $T_{bin}(\aut)$ if and only if there exists $r'=(q, g_r)\to(p, \bar{a}, push_2(A_{r_1}\ldots A_{r_k}))$ in $\Delta$.  
By induction, there exists exactly one $c'_1$ such that $T_{bin}^{-1}(c_1)=c'_1$, and hence exactly one $c'$ such that $T_{bin}^{-1}(c)=c'$.

Let $c$ be a symbolic well-formed run starting with $r=(q, g_r)\to(p, \bar{a}, push_1(A_{r_{k-1}}A_{r_k}))$.  
By construction and uniqueness of rules on added states:

\begin{align*}
&c = && q, A_1[stk_2]\Omega \\
&\xrightarrow{push_1(A_{r_{k-1}}A_{r_k})} && (q_{r_1}, A_{r_{k-1}}[stk_2] A_{r_k}[stk_2] \Omega) \\
&\xrightarrow{*} && (q_{r_i}, A_{r_i}[stk_2] A_{r_{i-1}}[stk_2] \ldots  A_{r_k}[stk_2] \Omega) \\
&\xrightarrow{*} && (q_{r_{k-1}}, A_{r_1}[stk_2] \ldots  A_{r_k}[stk_2] \Omega) \\
&\xrightarrow{id} && (p, A_1[stk_2] \ldots  A_{r_k}[stk_2] \Omega) \\
&\xrightarrow{T_{bin}(c_1)} && q_1, A_{r_2}[stk_2] \ldots  A_{r_k}[stk_2] \Omega_1 \\
&\xrightarrow{T_{bin}(c_2)} && \ldots  \\
&\xrightarrow{T_{bin}(c_{k-1})} && q_k, A_{r_k}[stk_2] \Omega_1 \\
&\xrightarrow{T_{bin}(c_k)} && q_r, \Omega_1
\end{align*}

By construction, $r$ exists in $T_{bin}(\aut)$ if and only if $r'=(q, g_r)\to(p, \bar{a}, push_1(A_{r_1}\ldots A_{r_k}))$ exists in $\Delta$.  
Hence, we can invert the run:

\begin{align*}
&T_{bin}^{-1}(c) = && q, A_1[stk_2]\Omega_1 \\
&\xrightarrow{r'} && p, A_{r_1}[stk_2] \ldots  A_{r_k}[stk_2] \Omega_1 \\
&\xrightarrow{T_{bin}^{-1}(c_1)} && q_1, A_{r_2}[stk_2] \ldots  A_{r_k}[stk_2] \Omega_1 \\
&\xrightarrow{T_{bin}^{-1}(c_2)} && \ldots  \\
&\xrightarrow{T_{bin}^{-1}(c_{k-1})} && q_k, A_{r_k}[stk_2] \Omega_1 \\
&\xrightarrow{T_{bin}^{-1}(c_k)} && q_r, \Omega_1
\end{align*}
By induction, there exists exactly one $c'_i$ such that $T_{bin}^{-1}(c_i)=c'_i$, and hence exactly one $c'$ such that $T_{bin}^{-1}(c)=c'$.
\end{proof}

\paragraph{Second transformation}

We now define a transformation $T_{prod}$ from the \emph{computing} 
$\stack^2(\struc)$-transducer $T_{bin}(\aut)$
to a \emph{computing} transducer satisfying the 
"only $pop_1$ can produce" property 
while preserving the "unary and binary $push_i$" property.

\begin{definition}
We define
$
T_{prod}(T_{bin}(\aut))=
(\mathcal{Q}_{prod},
T,
\Delta_{prod},
\stack(\Gamma_1, \stack(\Gamma_2, \struc)),
Z,
End)
$
where 
$\Delta_{prod}
=\bigcup_{r\in\Delta_{T_{bin}(\aut)}} T_{prod}(r)$.

For a rule 
$r=(q, \mathcal{P}_{A_1}\wedge g_r)\to(p, \bar a, op_r)$
in $\Delta_{T_{bin}(\aut)}$, we define:

\begin{align*}
&\text{if } \bar a=\varepsilon,
&& T_{prod}(r)=\{r\}, \\
&\text{if } op_r=pop_1,
&& T_{prod}(r)=\{r\}, \\
&\text{otherwise}, \\
&T_{prod}(r)=
\{(q, g_r)\to(q_r, \varepsilon, op_r)\}\\
&\qquad\cup
\{(q_r, \top)\to(q'_r, \varepsilon, push_1(A_1, A_1))\}\\
&\qquad\cup
\{(q'_r, \neg\mathcal{P}_{\#_2})
\to(q'_r, \varepsilon, pop_2)\}\\
&\qquad\cup
\{(q'_r, \mathcal{P}_{\#_2})
\to(p, \bar a, pop_1)\}.
\end{align*}
\end{definition}

\begin{lemma}
Let $\aut=(\mathcal{Q}, T, \Delta, \struc, Z, End)$
be a $\stack^2(\struc)$-transducer \emph{computing} $u(\sigma)$.
Then $T_{prod}(T_{bin}(\aut))$
is a $\stack^2(\struc)$-transducer computing $u(\sigma)$.
\end{lemma}

\begin{proof}
The proof is analogous to the proof for $T_{bin}$.
The construction only delays output production until a $pop_1$ transition is performed.
This delay is deterministic and finite, hence it preserves runs and the computed function.
\end{proof}

\paragraph{Third transformation.}
We finally define a transformation $T_{1symb}$ from the computing 
$\stack^2(\struc)$-transducer 
$T_{prod}(T_{bin}(\aut))$
to a computing transducer satisfying the 
``only one order $1$ symbol'' property 
while preserving the 
``unary and binary $push_i$'' and 
``only $pop_1$ can produce'' properties.

\begin{definition}
Let 
$\aut=
(\mathcal{Q}, T, \Delta,
\stack(\Gamma_1, \stack(\Gamma_2, \struc)),
Z, End)$
with 
$Z=\{A_1[A_2[\sigma]\#_2]\#_1
\mid \sigma\in\struc\}$,
and let 
$\aut_{prod}
= T_{prod}(T_{bin}(\aut))$.\\
We define
\begin{align*}
T_{1symb}(\aut_{prod})=
(\mathcal{Q}_{1symb},
T,
\Delta_{1symb},
\stack(\{A, \#_1\},
\stack((\Gamma_1\times\Gamma_2)\cup\Gamma_2,
\struc)),
Z_{1symb},
End),
\end{align*}
where 
$\Delta_{1symb}
=\bigcup_{r\in\Delta_{\aut_{prod}}}
T_{1symb}(r)$
and 
$Z_{1symb}
=\{A[(A_1, A_2)[\sigma]\#_2]\#_1
\mid \sigma\in\struc\}$.

For a rule
$r=(q, \mathcal{P}_{A_1}\wedge\mathcal{P}_{A_2}\wedge g_r)
\to(p, \bar a, op_r)$
in $\Delta_{\aut_{prod}}$, we define:

\begin{align*}
&\text{if } op_r\in\op_{\struc}, \\
&T_{1symb}(r)=
\{(q, \mathcal{P}_{(A_1, A_2)}\wedge g_r)
\to(p, \bar a, op_r)\}, \\
&\text{if } op_r=pop_2, \\
&T_{1symb}(r)=
\{(q, \mathcal{P}_{(A_1, A_2)}\wedge g_r)
\to(q_r, \bar a, pop_2)\}\\
&\qquad\cup
\{(q_r, \mathcal{P}_{B_2})
\to(p, \varepsilon,
push_2((A_1, B_2)))
\mid B_2\in\Gamma_2\}, \\
&\text{if } op_r=pop_1, \\
&T_{1symb}(r)=
\{(q, \mathcal{P}_{(A_1, A_2)}\wedge g_r)
\to(p, \bar a, pop_1)\}, \\
&\text{if } op_r=push_2(B_2, C_2), \\
&T_{1symb}(r)=
\{(q, \mathcal{P}_{(A_1, A_2)}\wedge g_r)
\to(p, \bar a,
push_2((A_1, B_2), C_2))\}, \\
&\text{if } op_r=push_1(B_1, C_1), \\
&T_{1symb}(r)=
\{(q, \mathcal{P}_{(A_1, A_2)}\wedge g_r)
\to(q_r, \varepsilon,
push_2((C_1, A_2)))\}\\
&\qquad\cup
\{(q_r, \top)
\to(q'_r, \varepsilon,
push_1(A, A))\}\\
&\qquad\cup
\{(q'_r, \top)
\to(p, \varepsilon,
push_2((B_1, A_2)))\}.
\end{align*}
\end{definition}

\begin{lemma}
Let $\aut$ be a computing 
$\stack^2(\struc)$-transducer computing $u(\sigma)$.
Then 
$T_{1symb}(T_{prod}(T_{bin}(\aut)))$
is a computing $\stack^2(\struc)$-transducer computing $u(\sigma)$.
\end{lemma}

\begin{proof}
This construction follows the encoding introduced by Fratani in Chapter 8 of her thesis~\cite{Fratani/thesis}, restricted to order $1$.
Her thesis provides a detailed correctness proof, which applies verbatim in our setting.
\end{proof}
\subsection{\cref{lem:bij_trans_to_trans_pred} proof}

Let $c$ be a symbolic well-formed-run of $\aut$ from $(q, stk.\Omega)$.\\

If the first rule $\delta$ of $c$ is a $pop_1$ rule then that is  the only rules of $c$ and $stk=X_1[\#_2]$.\\
If $\delta$ is non-productive then $c$ produce $a^0$. By definition of $T_2$ and the determinism of $\aut$, there is no rules, then $0$ accepting run, from $q, \#_2$.  \\
If $\delta$ is productive then $c$ produce $a^1$. By defininition of $T_2$ there exists an accepting run $(q, \#)\xrightarrow[T_2(\aut)]{T_2(\delta)} (\top, \#)$, since $T_2$ do not change the guard on $\#$.
By definition of $T_2$ and the determinism of $\aut$, there exists no other rules on $q, \#$. So, $T_2(\aut)$ produce $1$ accepting runs from $(q, \#)$\\

If the first rule $\delta$ of $c$ is a $push_1$ rules  then $c$  is \\
\begin{tabular}{l}
    $(q, X_1[stk]\Omega_1)\xrightarrow[\aut]{\delta}(p, X_1[stk]X_1[stk]\Omega_1)\xrightarrow[\aut]{c_1}(r, X_1[stk]\Omega_1)\xrightarrow[\aut]{c_2}(s, \Omega_1)$
\end{tabular}\\
Let say that $c_1$ produce $a^i$ and $c_2$ produce $a^j$. Then $c$ produce $a^{i+j}$.\\
By defininition of $T_2$, there exists rules $T_2(\delta)$:\\
\begin{tabular}{lll}
    &$\delta_1:$&$(q, \mathrm{g})\to(p, \varepsilon, id)$ \\
    $\forall r'\in \mathcal{Q},$&$\delta_{2_{r'}}:$&$(q, \mathrm{g}\wedge \mathcal{P}_{p, r'})\to(r', \varepsilon, id)$
\end{tabular}\\
By determinism of $\aut$ the predicates $\mathcal{P}{p, r'}(stk)$ is true if and only if $r=r'$. So there exists no accepting runs in $T_2(\aut)$ for $r\neq r'$.\\
By definition of $T_2$ and the determinism of $\aut$, there exists no other rules on $(q, stk)$.\\
By induction there is $i$ accepting runs from $(p, stk)$ and $j$ accepting runs from $(r, stk)$.\\
Since $T_2$ do not change the guard on $X[stk]$, there is $i+j$ runs from $(q, stk)$.\\

If the first rule $\delta$ of $c$ is an other type of rule then $T_2(\delta)=\delta$. Let say that $c$ produce $a^i$. By induction, there is $i$ accepting runs from $(q, stk)$.\\ 

\subsection{Complete \cref{def:transf_transpred_to_pred_gram}}
            We define $T_3(\aut_{\pred})=(N, \emptyset, \struc_{T_3(\aut_{\pred})}, S, P_{T_3(\aut_{\pred})})$ with
            \begin{itemize}
                \item $N=(\mathcal{Q}_{\pred}\times\Gamma\times \mathcal{Q}_{\pred}\times\mathcal{P} (\mathcal{Q}_{\pred}\times\mathcal{Q}_{\pred}))\cup \{S\}$ 
                \item $\struc_{T_3(\aut_{\pred})}=<\mathcal{S}, \mathcal{P}\cup\{\mathcal{P}_{M', M, X}|M', M\subseteq(\mathcal{Q}_{\pred}\times\mathcal{Q}_{\pred})\wedge X\in\Gamma\}, op>$
                \item \begin{tabular}{ll}
                    $P_{T_3(\aut_{\pred})}=$&$\{(S, \top)\to(((q_0, X, p, M),(p, \#, \top, \emptyset)), id)|p\in\mathcal{Q}_{\pred} \wedge M=\{(q, p)$\\
                    &$|((q, \mathcal{P}_{\#_2})\to(p, a, pop_1)\in \Delta)$\\
                    &$\vee(p=\bot \wedge \forall r,((q, \mathcal{P}_{\#_2})\to(r, a, pop_1)\not\in\Delta))\}\}\cup \underset{\delta\in\Delta_{\aut_{\pred}}}{\bigcup}T_3(\delta)$ 
                \end{tabular} 
            \end{itemize} 
            Where $T_3(\delta)$ is:\\ 
            \begin{tabular}{ll}
                If $\delta$ is $op_{2_{T_2}}$ then $T_3(\delta)$ =&$\{((q, Z, s, M ), \mathrm{g})\to((p, Z, s, M), op)$\\
                &$|s \in \mathcal{Q}\wedge M\subseteq(\mathcal{Q}\times\mathcal{Q})\}$\\
                If $\delta$ is $pop_{2_{T_2}}$ then $T_3(\delta)$ = &$\{((q, Z, p, M ), \mathrm{g})\to(\varepsilon)|M\subseteq(\mathcal{Q}\times \mathcal{Q})\}$\\
                If $\delta$ is $renaming_{2_{T_2}}$ then $T_3(\delta)$ = &$\{((q, Z, s, M ), \mathrm{g})\to((p, X, s, M), id)$\\
                &$|s \in \mathcal{Q}\wedge M\subseteq(\mathcal{Q}\times\mathcal{Q})\}$\\
                If $\delta$ is  $push_{2_{T_2}}$ then $T_3(\delta)$ = &$\{((q, Z, s, M), \mathrm{g}\wedge \mathcal{P}_{M', M, C})$\\
                &$\to((p, X, r, M')(r, Y, s, M), id)$\\
                &$|s\in\mathcal{Q} \wedge  p \in\mathcal{Q} \wedge M\subseteq(\mathcal{Q}\times\mathcal{Q})\}$\\
                If $\delta$ is productive $pop_{1_{T_2}}$ then $T_3(\delta)$ = &$((q, \#, \top, \emptyset), \mathrm{g})\to(\varepsilon)$\\
                If $\delta$ is second $push_{1_{T_2}}$ then $T_3(\delta)$ =&$\{((q, X, s, M), \mathrm{g}\wedge\mathcal{P}_{M', M, X})\to((r, X, s, M), id)$\\
                &$|s \in\mathcal{Q}\wedge M\subseteq(\mathcal{Q}\times\mathcal{Q})\wedge M\subseteq(\mathcal{Q}\times\mathcal{Q}) \wedge(p, r)\in M'\}$
            \end{tabular}\\
            With $M$ and $M'$ functional and total sets of pair of state, i.e. $\forall q\in\mathcal{Q}, \exists ! p\in\mathcal{Q}, (q, p)\in M$.\\    
            A predicate $\mathcal{P}_{M', M, Z} (\sigma)$ with $\sigma\in \mathcal{S}_\struc$ is true if and only if \\
            \begin{tabular}{l}
                $\forall stk\in \mathcal{S}_{\stack(\struc)}, (M=\{(q, p)\mid q, X_1[stk]\Omega_1 \underset{\aut}{\xrightarrow{*}}p, \Omega_1 \})\implies$\\
                $M'=\{(q, p)| \exists r, p\in M, (q, X_1[Z_2[\sigma]res]\Omega_1)\underset{\aut}{\xrightarrow{*}} (r, X_1[res]\Omega_1)\}$.
            \end{tabular}

\subsection{\cref{lem:bij_trans_to_alt_trans} proof}
As said in the proof sketch we prove both direction separetly.
\paragraph*{direction 1: $T_4$ is a total function}
Let $c$ be a symbolic well-formed-run of $\aut$ from $(q, X_1[stk].\Omega)$.\\

If the first rule $\delta$ of $c$ is a $pop_1$ rule then \\
\begin{tabular}{ll}
    $c=(q, X_1[\#_2]\Omega_1)\xrightarrow[\aut]{\delta}(p, \Omega_1)$
\end{tabular}\\
By definition of $T_4$, there exists an accepting run tree $((q, p), \#_2)\to(\top, \#_2)$.
By the determinism of $\aut$ and definition of $T_4$, there exists no other rules on $((q, p), \#_2)$ and there is no accepting run tree from $((q, p'), \#_2)$ for $p'\neq p $.\\

If the first rule $\delta$ of $c$ is a $push_1$ rule then \\
\begin{tabular}{ll}
    $c=(q, X_1[stk]\Omega_1)\xrightarrow[\aut]{\delta}(p, X_1[stk]X_1[stk]\Omega_1)$\\
    $\xrightarrow[\aut]{c_1}(r, X_1[stk]\Omega_1)\xrightarrow[\aut]{c_2}(s, \Omega_1)$
\end{tabular}\\
By definition of $T_4$, there exists rules \\
\begin{tabular}{ll}
    $\forall s', r' \in \mathcal{Q}$$((q, s'), \mathrm{g})\to((p, r'), \varepsilon, id)\wedge ((r', s'), \varepsilon, id)$
\end{tabular}\\
By induction, there exists a run tree \\
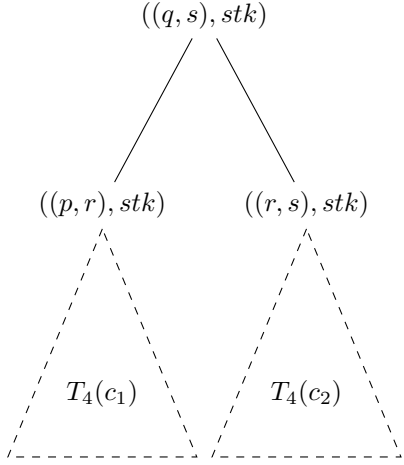
\begin{tikzpicture}[ level distance=2.5cm, sibling  distance=2.7cm]
    \node {$((q, s), stk)$}
        child {
            node {$((p, r), stk)$ }
            child{
                node [draw, dashed, shape border uses incircle, isosceles triangle,
                                        shape border rotate=90, minimum height=10mm] {$T_4(c_1)$}
            }
        }
        child {
            node {$((r, s), stk)$ }
            child{
                node [draw, dashed, shape border uses incircle, isosceles triangle,
                                        shape border rotate=90, minimum height=10mm] {$T_4(c_2)$}
            }
    };
\end{tikzpicture}\\
By induction, there exists no accepting run trees from $((p, r'), stk)$ for $r'\neq r$ and no accepting run trees from $((r, s'), stk)$ for $s'\neq s$.\\
Then,  there is no accepting run trees starting by rules obtained by $T_4(\delta)$ for $r'\neq r$ or $s'\neq s$.
 By the determinism of $\aut$ and definition of $T_4$, there exists no other rules on $((q, s'), stk)$ for all $s'\in \mathcal{Q}$.\\

 For other type of rules, $T_4$ is the same. So, we just make the proof for $operation$ rules as an example.\\
 
 If the first rule $\delta$ of $c$ is a $operation$ rule then \\
\begin{tabular}{ll}
    $c=(q, X_1[stk]\Omega_1)\xrightarrow[\aut]{\delta}(p, X_1[op(stk)]\Omega_1)$\\
    $\xrightarrow[\aut]{c_1}(s, \Omega_1)$
\end{tabular}\\
By definition of $T_4$, there exists rules \\
\begin{tabular}{ll}
    $\forall s \in \mathcal{Q}$$((q, s'), \mathrm{g})\to((p, s'), \varepsilon, op)$
\end{tabular}\\
By induction, there exists a run tree \\
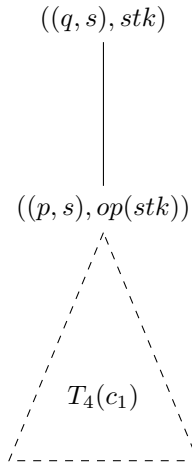
\begin{tikzpicture}[ level distance=2.5cm, sibling  distance=2.7cm]
    \node {$((q, s), stk)$}
        child {
            node {$((p, s), op(stk))$ }
            child{
                node [draw, dashed, shape border uses incircle, isosceles triangle,
                                        shape border rotate=90, minimum height=10mm] {$T_4(c_1)$}
            }
        }
   ;
\end{tikzpicture}\\
By induction, there exists  no accepting run trees from $((p, s'), op(stk))$ for $s'\neq s$.\\
Then,  there is no accepting run trees starting by rules obtained by with $T_4(\delta)$ for $s'\neq s$.
 By the determinism of $\aut$ and definition of $T_4$, there exists no other rules on $((q, s'), stk)$ for all $s'\in \mathcal{Q}$.\\

\paragraph*{direction 2: ${T_4}^{-1}$ is a total function}

Let $\alpha$ be an accepting run tree of $\aut$ from $((q, s), stk)$.\\

If the first rule $\delta$ of $\alpha$ is $((q, s), \mathcal{P}_{\#})\to(\top, \varepsilon, id)$ then \\
\begin{tabular}{ll}
    $\alpha=((q, s), \#)\xrightarrow[T_4(\aut)]{\delta}(\top, \#)$
\end{tabular}\\
By definition of $T_4$, there exists $\delta\in \Delta_{T_4(\aut)}$ if and only if there exists a rule $(q, \mathcal{P}_\#)\to(s, \varepsilon, pop_1) \in \Delta_\aut$.
So, there exists a symbolic well-formed run  $(q, X_1[\#_2]\Omega_1)\xrightarrow[\aut]{{T_4}^{-1}(\delta)}(s, \Omega_1)$.\\
By determinism of $\aut$ , there exists no other rules on $(q, X_1[\#_2]\Omega_1)$.\\

If the first rule $\delta$ of $\alpha$ is $((q, s), \mathrm{g})\to((p, r), \varepsilon, id)\wedge ((r, s), \varepsilon, id)$ rule then $\alpha$ is \\
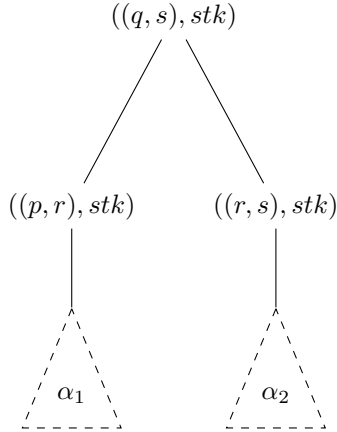
\begin{tikzpicture}[ level distance=2.5cm, sibling  distance=2.7cm]
    \node {$((q, s), stk)$}
        child {
            node {$((p, r), stk)$ }
            child{
                node [draw, dashed, shape border uses incircle, isosceles triangle,
                                        shape border rotate=90, minimum height=10mm] {$\alpha_1$}
            }
        }
        child {
            node {$((r, s), stk)$ }
            child{
                node [draw, dashed, shape border uses incircle, isosceles triangle,
                                        shape border rotate=90, minimum height=10mm] {$\alpha_2$}
            }
    };
\end{tikzpicture}\\
By definition of $T_4$, there exists $\delta\in \Delta_{T_4(\aut)}$ if and only if there exists a rule $(q, \mathrm{g}) \to (p, \varepsilon, push_1(X_1 X_1))\in \aut$.\\
By induction, there exists a symbolic well-formed run\\
\begin{tabular}{ll}
    $(q, X_1[stk]\Omega_1)\xrightarrow[\aut]{{T_4}^{-1}(\delta)}(p, X_1[stk]X_1[stk]\Omega_1)$\\
    $\xrightarrow[\aut]{{T_4}^{-1}(\alpha_1)}(r, X_1[stk]\Omega_1)\xrightarrow[\aut]{{T_4}^{-1}(\alpha_2)}(s, \Omega_1)$
\end{tabular}\\
By determinism of $\aut$ , there exists no other rules on $(q, X_1[stk]\Omega_1)$.\\

 For other type of rules, $T_4$ is the same. So, we just make the proof for $operation$ rules as an example.\\
 
 If the first rule $\delta$ of $\alpha$ is $((q, s), \mathrm{g})\to((p, s), \varepsilon, op)$ rule then $\alpha$ is \\
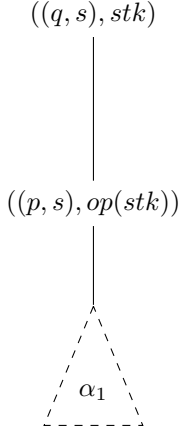
\begin{tikzpicture}[ level distance=2.5cm, sibling  distance=2.7cm]
    \node {$((q, s), stk)$}
        child {
            node {$((p, s), op(stk))$ }
            child{
                node [draw, dashed, shape border uses incircle, isosceles triangle,
                                        shape border rotate=90, minimum height=10mm] {$\alpha_1$}
            }
        }
   ;
\end{tikzpicture}\\
By definition of $T_4$, there exists $\delta\in \Delta_{T_4(\aut)}$ if and only if there exists a rule $(q, \mathrm{g}) \to (p, \varepsilon, op)\in \Delta_\aut$.\\
By induction, there exists a symbolic well-formed run\\
\begin{tabular}{ll}
    $(q, X_1[stk]\Omega_1)\xrightarrow[\aut]{{T_4}^{-1}(\delta)}(p, X_1[op(stk)]\Omega_1)$\\
    $\xrightarrow[\aut]{{T_4}^{-1}(\alpha_1)}(s, \Omega_1)$
\end{tabular}\\
By determinism of $\aut$ , there exists no other rules on $(q, X_1[stk]\Omega_1)$.

\subsection{\cref{lem:alt_trans_unic} proof}
Let $q\in\mathcal{Q}$, $u\in \mathcal{S}_{\mathcal{P}(\aut)}$ and $\psi\subseteq \mathcal{P}(\mathcal{Q}\times \mathcal{Q})$ a functional set.\\
Let $\alpha$ a run tree $(q, p), stk. \Omega\xrightarrow[T_4(\aut)]{\alpha} \underset{(q_i, p_i)\in\varphi}{\bigwedge}(q_i, p_i), \Omega$ with $\varphi\subseteq \psi$.\\
We prove that there is not $s'\neq s$, $\varphi'\neq\varphi$ and $\alpha'$ such that  $(q, p'), stk. \Omega\xrightarrow[T_4(\aut)]{\alpha'} \underset{(q_i, p_i)\in\varphi'}{\bigwedge}(q_i, p_i), \Omega$ with $\varphi'\subseteq \psi$.\\

If the first rule  $\delta$ of $\alpha$ is $((q, s), \mathrm{g})\to((p, r), \varepsilon, id)\wedge ((r, s), \varepsilon, id)$ then $\alpha$ is \\
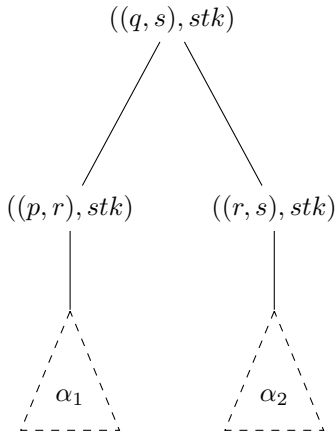
\begin{tikzpicture}[ level distance=2.5cm, sibling  distance=2.7cm]
    \node {$((q, s), stk)$}
        child {
            node {$((p, r), stk)$ }
            child{
                node [draw, dashed, shape border uses incircle, isosceles triangle,
                                        shape border rotate=90, minimum height=10mm] {$\alpha_1$}
            }
        }
        child {
            node {$((r, s), stk)$ }
            child{
                node [draw, dashed, shape border uses incircle, isosceles triangle,
                                        shape border rotate=90, minimum height=10mm] {$\alpha_2$}
            }
    };
\end{tikzpicture}\\
We note  $(p, r), stk.\Omega\xrightarrow[T_4(\aut)]{\alpha_1} \underset{(q_i, p_i)\in\varphi_1}{\bigwedge}(q_i, p_i), \Omega$ \\and $(r, s), stk.\Omega\xrightarrow[T_4(\aut)]{\alpha_2} \underset{(q_i, p_i)\in\varphi_2}{\bigwedge}(q_i, p_i), \Omega$.

By defininition of $T_4$, there exists, for all $r' \in \mathcal{Q}$, rule:
\begin{tabular}{ll}
    $((q, s), \mathrm{g})\to((p, r'), \varepsilon, id)\wedge ((r', s), \varepsilon, id)$
\end{tabular}\\
By induction, for all $r'\neq r$ there is no $\varphi'_1\subseteq \psi$ such that $(p, r'), stk.\Omega\xrightarrow[T_4(\aut)]{\alpha'_1} \underset{(q_i, p_i)\in\varphi'_1}{\bigwedge}(q_i, p_i), \Omega$.\\
By induction,  there is no $\varphi_1 \neq \varphi'_1$ such that $(p, r), stk.\Omega\xrightarrow[T_4(\aut)]{\alpha'_1} \underset{(q_i, p_i)\in\varphi'_1}{\bigwedge}(q_i, p_i), \Omega$ with $\varphi'_1\subseteq \psi$.\\
and there is no $\varphi_2 \neq \varphi'_2$ such that $(r, s), stk.\Omega\xrightarrow[T_4(\aut)]{\alpha'_2} \underset{(q_i, p_i)\in\varphi'_2}{\bigwedge}(q_i, p_i), \Omega$ with $\varphi'_2\subseteq \psi$.\\
So, there is no $\varphi \neq \varphi'$ such that $(q, s), stk.\Omega\xrightarrow[T_4(\aut)]{\alpha'} \underset{(q_i, p_i)\in\varphi'}{\bigwedge}(q_i, p_i), \Omega$ with $\varphi'\subseteq \psi$.

For $s'\neq s$ there is a rule $(q, s'), stk\xrightarrow[T_4(\aut)]((p, r), \varepsilon, id)\wedge ((r, s'), \varepsilon, id)$.\\
By induction, there is no $\varphi'_2\subseteq \psi$ such that $(r, s'), stk.\Omega\xrightarrow[T_4(\aut)]{\alpha'_2} \underset{(q_i, p_i)\in\varphi'_2}{\bigwedge}(q_i, p_i), \Omega$.\\
So, there is no $\varphi \neq \varphi'$ or $s'\neq s$ such that $(q, s'), stk.\Omega\xrightarrow[T_4(\aut)]{\alpha'} \underset{(q_i, p_i)\in\varphi'}{\bigwedge}(q_i, p_i), \Omega$ with $\varphi'\subseteq \psi$.\\

If the first rule  $\delta$ of $\alpha$ is an other type of rule then $\alpha$ is \\
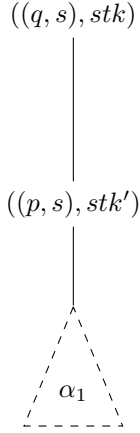
\begin{tikzpicture}[ level distance=2.5cm, sibling  distance=2.7cm]
    \node {$((q, s), stk)$}
        child {
            node {$((p, s), stk')$ }
            child{
                node [draw, dashed, shape border uses incircle, isosceles triangle,
                                        shape border rotate=90, minimum height=10mm] {$\alpha_1$}
            }
        }
   ;
\end{tikzpicture}\\
By induction, there exists no $\varphi'\neq \varphi$ or $s'\neq s$  $(p, s'), stk'.\Omega\xrightarrow[T_4(\aut)]{\alpha_1} \underset{(q_i, p_i)\in\varphi'}{\bigwedge}(q_i, p_i), \Omega$ with $\varphi'\subseteq \psi$. \\
By definition of $T_4$ and determinism of $\aut$, the only rule on $(q, s'), stk.\Omega$ is $((q, s'), \mathrm{g})\to ((p, s'), \varepsilon, op)$.\\
Consequently, there is no $\varphi \neq \varphi'$ or $s'\neq s$ such that $(q, s'), stk.\Omega\xrightarrow[T_4(\aut)]{\alpha'} \underset{(q_i, p_i)\in\varphi'}{\bigwedge}(q_i, p_i), \Omega$ with $\varphi'\subseteq \psi$.\\
\subsection{\cref{lem:bij_alt_trans_to_proof_gram} proof}
 As said the proof sketch, we consider a promotion $T_5$ as a relation from the symbolic well-formed run trees of $\aut_{\alt}$ to the derivation trees of $T_5(\aut_{\alt})$ rooted in $(X,(q, p)\to\varphi)[\sigma]$.
The relation $T_5$ associates symbolic well-formed run tree $(q, p), X[\sigma]\Omega \xrightarrow{c} \varphi, \Omega$ with derivation  trees from $(X,(q, p)\to\varphi)[\sigma]$.\\

\paragraph*{direction 2: ${T_5}^{-1}$ is a total function}
    As stated in the proof sketch, the functionality of ${T_5}^{-1}$ is direct since there is at most one symbolic well-formed run tree from $q, X[\sigma]$ to $\varphi$ functional as established in \cref{lem:alt_trans_unic}.  \\

    We prove that ${T_5}^{-1}$ is total.\\
    Let $t$ be a derivation tree of $T_5(\aut_\alt)$ from $(Z,(q, s)\to\varphi)[\sigma]$.

    If $t$ is \\
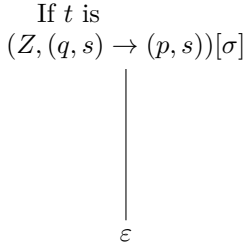
    \begin{tikzpicture}[ level distance=2.5cm, sibling  distance=2.7cm]
        \node {$(Z,(q, s)\to (p, s))[\sigma]$}
            child {
                node {$\varepsilon$ }   
        };
    \end{tikzpicture}\\
    Then, there is a rule $r=(((Z,(q, s)\to(p, s)), g) \to (\varepsilon, id)) \in P_{T_5(\aut_{alt})}$.\\
    By definition of $T_5$, there is such rule if and only if there is a rule $((q, s), g)\to((p, s), \varepsilon, id)\in \Delta_{Alt}$. So, there is a well-formed run tree $(q, p), Z[\sigma]\Omega)\xrightarrow[\aut_{alt}]{{T_5}^{-1}(r)} (\top, \#)$. \\

    If $t$ is \\
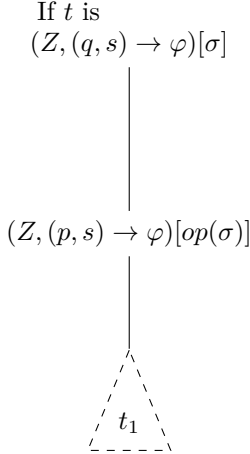
    \begin{tikzpicture}[ level distance=2.5cm, sibling  distance=2.7cm]
        \node {$(Z,(q, s)\to\varphi)[\sigma]$}
            child {
            node {$(Z,(p, s)\to\varphi)[op(\sigma)]$ }
            child{
                node [draw, dashed, shape border uses incircle, isosceles triangle,
                                        shape border rotate=90, minimum height=10mm] {$t_1$}
            }
        };
    \end{tikzpicture}\\
    Then, there is a rule $r=((Z,((q, s)\to\varphi)), g)\to((Z,((p, s)\to\varphi)), op)) \in P_{T_5(\aut_{alt})}$.\\
    By definition of $T_5$, there is such rule if and only if there is a rule $((q, s), g)\to((p, s), \varepsilon, op)\in \Delta_{Alt}$.\\
    By induction on $t_1$, there is a symbolic well-formed run $((p, s), Z[op(\sigma)]\Omega)\xrightarrow[\aut_{\alt}]{{T_5}^{-1}(t_1)}((\varphi, \Omega)$.\\
    So, there is a symbolic well-formed run tree $((q, s), Z[\sigma]\Omega)\xrightarrow[\aut_{\alt}]{{T_5}^{-1}(r)}((p, s), Z[op(\sigma)]\Omega)\xrightarrow[\aut_{\alt}]{{T_5}^{-1}(t_1)}(\varphi, \Omega)$.\\

    If $t$ is \\
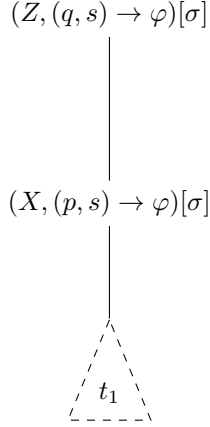
    \begin{tikzpicture}[ level distance=2.5cm, sibling  distance=2.7cm]
        \node {$(Z,(q, s)\to\varphi)[\sigma]$}
            child {
            node {$(X,(p, s)\to\varphi)[\sigma]$ }
            child{
                node [draw, dashed, shape border uses incircle, isosceles triangle,
                                        shape border rotate=90, minimum height=10mm] {$t_1$}
            }
        };
    \end{tikzpicture}\\
    Then, there is a rule $r=((Z,((q, s)\to\varphi)), g)\to((X,((p, s)\to\varphi)), id)) \in P_{T_5(\aut_{alt})}$.\\
    By definition of $T_5$, there is such rule if and only if there is a rule $((q, s), g)\to((p, s), \varepsilon, push_2(X))\in \Delta_{Alt}$.\\
    By induction on $t_1$, there is a symbolic well-formed run $((p, s), X[\sigma]\Omega)\xrightarrow[\aut_{\alt}]{{T_5}^{-1}(t_1)}((\varphi, \Omega)$.\\
    So, there is a symbolic well-formed run tree $((q, s), Z[\sigma]\Omega)\xrightarrow[\aut_{\alt}]{{T_5}^{-1}(r)}((p, s), X[\sigma]\Omega)\xrightarrow[\aut_{\alt}]{{T_5}^{-1}(t_1)}(\varphi, \Omega)$.\\

    If $t$ is \\
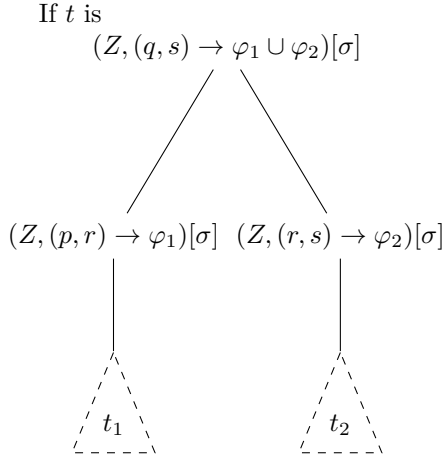
    \begin{tikzpicture}[ level distance=2.5cm, sibling  distance=3cm]
        \node {$(Z,(q, s)\to\varphi_1\cup\varphi_2)[\sigma]$}
            child {
                node {$(Z,(p, r)\to\varphi_1)[\sigma]$ }
                child{
                    node [draw, dashed, shape border uses incircle, isosceles triangle,
                                            shape border rotate=90, minimum height=10mm] {$t_1$}
                }
            }
        child {
            node {$(Z,(r, s)\to\varphi_2)[\sigma]$ }
            child{
                node [draw, dashed, shape border uses incircle, isosceles triangle,
                                        shape border rotate=90, minimum height=10mm] {$t_2$}
            }
        };
    \end{tikzpicture}\\
    Then, there is a rule $r=(Z,(((q, s))\to(\varphi_1\cup\varphi_2)), g)\to$\\$((Z, ((p, r)\to \varphi_1))(Z, ((r, s)\to \varphi_2)), id)\in P_{T_5(\aut_{alt})}$.\\
    By definition of $T_5$, there is such rule if and only if there is a rule $((q, s), \mathrm{g})\to((p, r), \varepsilon, id)\wedge ((r, s), \varepsilon, id)\in \Delta_{Alt}$.\\
    By induction on $t_1$, there is a symbolic well-formed run $((p, r), Z[\sigma]\Omega)\xrightarrow[\aut_{\alt}]{{T_5}^{-1}(t_1)}((\varphi_1, \Omega)$.\\
    By induction on $t_2$, there is a symbolic well-formed run $((r, s), Z[\sigma]\Omega)\xrightarrow[\aut_{\alt}]{{T_5}^{-1}(t_2)}((\varphi_2, \Omega)$.\\
    So, there is a symbolic well-formed run tree $((q, s), Z[\sigma]\Omega)\xrightarrow[\aut_{\alt}]{{T_5}^{-1}(r)}((p, r), X[\sigma]\Omega)\wedge ((r, s), X[\sigma]\Omega)$\\
    $\xrightarrow[\aut_{\alt}]{{T_5}^{-1}(t_1)\wedge {T_5}^{-1}(t_2)}(\varphi_1\cup\varphi_2, \Omega)$.\\

    If $t$ is \\
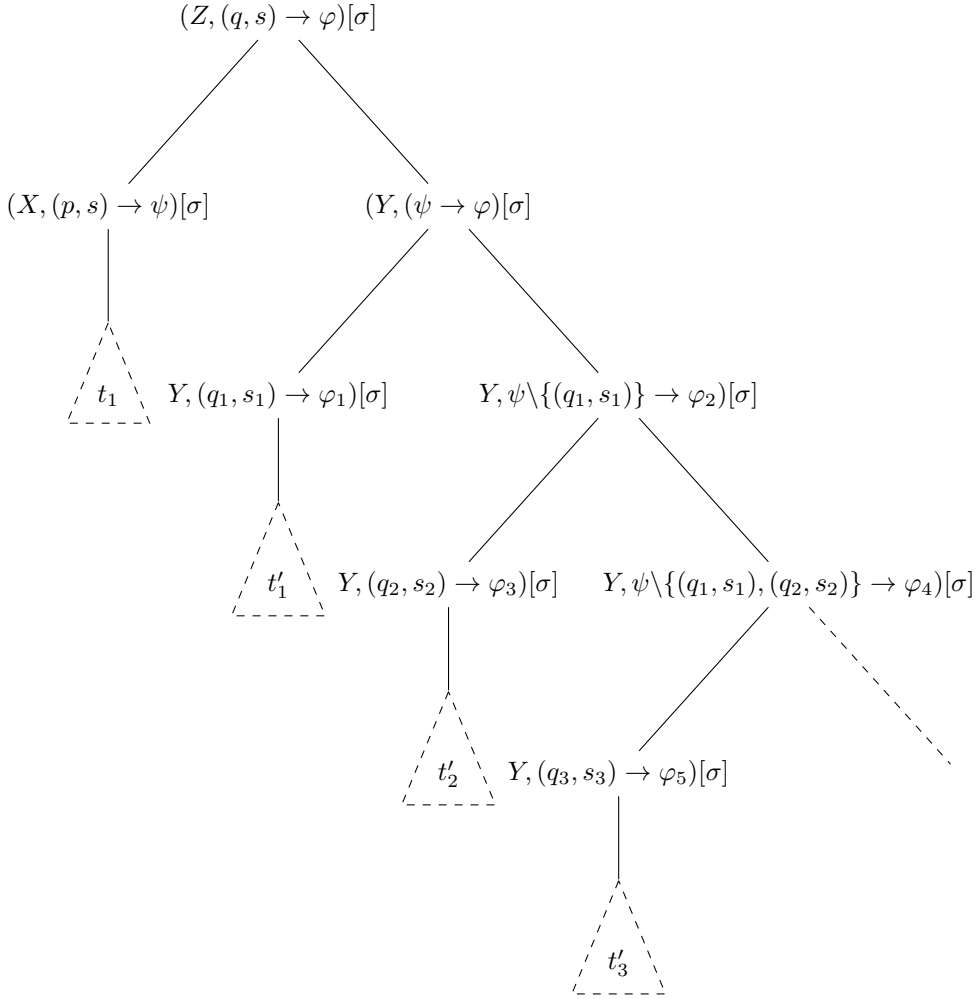
    \begin{tikzpicture}[ level distance=2.5cm, sibling  distance=4.5cm]
        \node {$(Z,(q, s)\to\varphi)[\sigma]$}
            child {
                node {$(X,(p, s)\to\psi)[\sigma]$ }
                child{
                    node [draw, dashed, shape border uses incircle, isosceles triangle,
                                            shape border rotate=90, minimum height=10mm] {$t_1$}
                }
            }
        child {node {$(Y,(\psi\to\varphi)[\sigma]$ }
            child {
                node {$Y,(q_1, s_1)\to\varphi_1)[\sigma]$ }
                child{
                    node [draw, dashed, shape border uses incircle, isosceles triangle,
                                            shape border rotate=90, minimum height=10mm] {$t'_1$}
                }
            }
            child {
                node {$Y, \psi\backslash\{(q_1, s_1)\}\to\varphi_2)[\sigma]$ }
                child {
                    node {$Y,(q_2, s_2)\to\varphi_3)[\sigma]$ }
                    child{
                        node [draw, dashed, shape border uses incircle, isosceles triangle,
                                                shape border rotate=90, minimum height=10mm] {$t'_2$}
                    }
                }
                child {
                    node {$Y, \psi\backslash\{(q_1, s_1),(q_2, s_2)\}\to\varphi_4)[\sigma]$ }
                    child {
                    node {$Y,(q_3, s_3)\to\varphi_5)[\sigma]$ }
                        child{
                            node [draw, dashed, shape border uses incircle, isosceles triangle,
                                                    shape border rotate=90, minimum height=10mm] {$t'_3$}
                        }
                    }child[dashed]{
                        node {}
                    }
                }
            }
        };
    \end{tikzpicture}\\
    Then, there is a rule $r=(Z,(((q, s))\to\varphi), g)\to$\\
    $((X, ((p, s)\to \psi))(Y, (\psi\to \varphi)), id)\in P_{T_5(\aut_{alt})}$.\\
    By definition of $T_5$, there is such rule if and only if there is a rule $((q, s), \mathrm{g})\to((p, s), \varepsilon, push_2)\in \Delta_{Alt}$.\\
    And, there is a rule $r_i=(Y,(((\psi)\to\varphi), g)\to$\\
    $((Y, ((q_i, s_i)\to \varphi_{2i-1}))(Y, (\psi\backslash\{(q_j, s_j)|j\leq i\}\to \varphi_{2i})), id)\in P_{T_5(\aut_{alt})}$.\\
    By induction on $t_1$, there is a symbolic well-formed run $((p, s), X[\sigma]\Omega)\xrightarrow[\aut_{\alt}]{{T_5}^{-1}(t_1)}((\psi, \Omega)$.\\
    By induction on $t'_1$, there is a symbolic well-formed run $(q_i, s_i), Y[\sigma]\Omega)\xrightarrow[\aut_{\alt}]{{T_5}^{-1}(t_i)}((\varphi_{2i+1}, \Omega)$.\\
    So, there is a symbolic well-formed run tree $((q, s), Z[\sigma]\Omega)\xrightarrow[\aut_{\alt}]{{T_5}^{-1}(r)}((p, r), X[\sigma]Y[\sigma]\Omega)$\\
    $\xrightarrow[\aut_{\alt}]{ {T_5}^{-1}(t_1)}(\psi, Y[\sigma]\Omega)$\\
    $\xrightarrow[\aut_{\alt}]{ \bigwedge{T_5}^{-1}(t'_i)}(\varphi, \Omega)$.\\

    If $t$ is \\
    \begin{tikzpicture}[ level distance=2.5cm, sibling  distance=2.7cm]
        \node {$(\#,(q, s)\to\top)[\sigma]$}
            child {
                node {$\varepsilon$ }   
        };
    \end{tikzpicture}\\
    Then, there is a rule $(((\#,(q, s)\to\top), \top) \to (\varepsilon, id)) \in P_{T_4(\aut_{alt})}$.\\
    By definition of $T_4$, there is such rule if and only if there is a rule $((q, p), \mathcal{P}_{\#})\to(\top, \varepsilon, id)\in \Delta_{Alt}$. So, there is an accepting run $(q, p), \#)\to (\top, \#)$. \\

    \begin{tabular}{ll}
                If $\delta$ is a $pop_1$ rule then $T_4(\delta)$ =&$\{((q, p), \mathcal{P}_{\#})\to(\top, \varepsilon, id)\}$\\
                If $\delta$ is a $push_1$ rule then $T_4(\delta)$ =&$\{((q, s), \mathrm{g})\to((p, r), \varepsilon, id)\wedge ((r, s), \varepsilon, id)$\\
                &$|r\in\mathcal{Q}\wedge s\in\mathcal{Q}\}$\\
                Else $T_4(\delta)$ =&$\{((q, s), \mathrm{g})\to((p, s), \varepsilon, op_r)|s\in\mathcal{Q}\}$
            \end{tabular}

                \begin{tabular}{ll}
                    $Decomp$=&$\{((X,((q, p)\cup\psi)\to(\varphi_1\cup\varphi_2)), \top)$\\
                    &$\to((X, ((q, p)\to \varphi_1))(A, (\psi\to \varphi_2)), id)$\\
                    &$|\psi, \varphi_1, \varphi_2\in \mathcal{P}(\mathcal{Q}_{\alt}), \forall v\in\psi, (q_i, q_j)<v \}$\\
                    &$|\varphi\in \mathcal{P}(\mathcal{Q}_{\alt})\}$\\
                    If $r$ is  $push_{2_{T_4}}$ then $T_5(r)$ =&$\{(Z,(((q, s))\to\varphi), g)\to$\\
                    &$((X, ((p, s)\to \psi))(Y, (\psi\to \varphi)), id)$\\
                    &$|\psi, \varphi\subseteq\mathcal{P} (\mathcal{Q}_{\alt}\text{ functional })\}$\\
                \end{tabular}     
\paragraph*{Technical lemma}
    We prove by a structural induction on derivation tree  that there exists at most one derivation tree rooted in  $(A,(q, s)\to\varphi)$ with $\varphi$ functional.\\ 
    Here is an example for rules obtained by $T_1$ on $push_{1_{\alt}}$ rules.\\
    We take a derivation tree $t$ from $(A,(q, p)\to\varphi)$ starting by a a rule $r$ .\\
    If $t$ is:   \\
    \begin{tikzpicture}[ level distance=2.5cm, sibling  distance=2.7cm]
        \node {$(Z,(q, s)\to\varphi)[\sigma]$}
            child {
            node {$(Z,(p, r)\to\varphi_1)[\sigma]$ }
            child{
                node [draw, dashed, shape border uses incircle, isosceles triangle,
                                        shape border rotate=90, minimum height=10mm] {$t_1$}
                }
            }child {
            node {$(Z,(r, s)\to\varphi_2)[\sigma]$ }
            child{
                node [draw, dashed, shape border uses incircle, isosceles triangle,
                                        shape border rotate=90, minimum height=10mm] {$t_2$}
                }
            }
        
        ;
    \end{tikzpicture}\\
    then $r=((Z,((q,s)\to\varphi)),g)\to ((Z,((p,r)\to\varphi_1))(Z,((r,s)\to\varphi_1)),id)$.\\
    By \cref{lem:alt_trans_unic},  \\
    \begin{tabular}{ll}
        $\exists^{\leq1}r\in\mathcal{Q}_\aut, $&$\varphi_1\subseteq\varphi, \alpha\in\mathcal{T} (\Delta_{\aut_\alt}), (p, r)Z[\sigma]\Omega\underset{\aut_\alt}{\xrightarrow{\alpha}}\underset{(q_i, p_i)\in\varphi_1}{\bigwedge} (q_i, p_i), \Omega$.\\
        &$\varphi_2\subseteq\varphi, \alpha\in\mathcal{T} (\Delta_{\aut_\alt}), (r, s)Z[\sigma]\Omega\underset{\aut_\alt}{\xrightarrow{\alpha}}\underset{(q_i, p_i)\in\varphi_2}{\bigwedge} (q_i, p_i), \Omega$.\\
    \end{tabular}\\
    So,by the the totallity of ${T_5}^{-1}$ ,  there no derivation trees $(Z,(p, r')\to\varphi'_1)[\sigma]$ and $(Z,(r, s)\to\varphi'_2)[\sigma]$ for  $r\neq r'$, $\varphi'_1\neq \varphi_1$ and $\varphi'_2\neq \varphi_2$. \\
    By induction, $t_1$ is the only derivation tree from $(Z,(p, r)\to\varphi)[\sigma]$ and $t_2$ is the only derivation tree from $(Z,(r, s)\to\varphi)[\sigma]$.\\
    Then, $t$ is the only derivation tree from $(Z,(q, s)\to\varphi)[\sigma]$.\\

    For each case of this induction, we proceed similarly.\\
    There is a minor difficulty for  $push_{2_{T_4}}$ due to the add of $\psi$. Which is solved by  the functionalty of $\psi$ allowing the use of \cref{lem:alt_trans_unic} and the order $>$ on state of $\aut_{alt}$, ensuring the uniqueness on the separation made by $Decomp$ ans $\psi$. 
\paragraph*{direction 2: ${T_5}$ is a total function}
    As stated in the proof sketch, the functionality of ${T_5}$ is direct since there is at most one derivation tree rooted in  $(A,(q, s)\to\varphi)$ with $\varphi$ functional.  \\
    
    We prove that ${T_5}$ is total by structural induction on well formed run trees of $\aut_\alt$ similarly to the proof of totality for ${T_5}^{-1}$.
\end{document}